\documentclass[smallextended]{svjour3}       % onecolumn (second format)
\smartqed  % flush right qed marks, e.g. at end of proof
\usepackage{graphicx,color,amsmath,amssymb}
\usepackage[percent]{overpic}

\newcommand{\pps}{p^{+}}
\newcommand{\pms}{p^{-}}
\newcommand{\gradS}{\frac{\partial S}{\partial x}}
\newcommand{\Der}{\mathrm{d}}
\newcommand{\R}{\mathbb{R}}
\newcommand{\pd}[2]{\dfrac{\partial#1}{\partial#2}}
\newcommand{\re}[1]{\mbox{$($\ref{#1}$)$}}
\newcommand{\ds}{\displaystyle}
\newcommand{\bv}{\mathbf{v}}
\newcommand{\dd}[2]{\dfrac{\mathrm{d} #1}{ \mathrm{d} #2}}
\newcommand{\q}{~\mbox{\hspace{2em}\frame{\rule{0ex}{1.5ex}\mbox{\hspace{1ex}}}}\vspace*{2ex}}
\numberwithin{equation}{section}
\newtheorem{thm}{Theorem}

\journalname{Submitted to Bulletin of Mathematical Biology}

\begin{document}
\title{Travelling waves in hybrid chemotaxis models}
\author{Benjamin Franz \and Chuan Xue \and Kevin J. Painter \and Radek Erban}

\authorrunning{Franz, Xue, Painter, Erban} % if too long for running head

\institute{Benjamin Franz \and Radek Erban \at
Mathematical Institute, University of Oxford \\
24-29 St. Giles', Oxford, OX1 3LB, United Kingdom \\
\email{franz@maths.ox.ac.uk, erban@maths.ox.ac.uk}
\and
Chuan Xue \at
Department of Mathematics, Ohio State University \\
231 West 18th Avenue, Columbus, OH 43210, USA \\
\email{cxue@math.osu.edu}
\and
Kevin J. Painter \at
Department of Mathematics, Heriot-Watt University \\
Edinburgh, EH14 4AS, United Kingdom \\
\email{K.J.Painter@ma.hw.ac.uk}
}

\date{Preprint version: \today}

\maketitle

\begin{abstract}
Hybrid models of chemotaxis combine agent-based models of cells
with partial differential equation models of extracellular
chemical signals. In this paper, travelling wave properties of
hybrid models of bacterial chemotaxis are investigated. Bacteria
are modelled using an agent-based (individual-based) approach with
internal dynamics describing signal transduction. In addition to
the chemotactic behaviour of the bacteria, the indivi\-dual-based
model also includes cell proliferation and death. Cells consume
the extracellular nutrient field (chemoattractant) which is
modelled using a partial differential equation. Mesoscopic and
macroscopic equations representing the behaviour of the hybrid
model are derived and the existence of travelling wave solutions
for these models is established. It is shown that cell
proliferation is necessary for the existence of non-transient
(stationary) travelling waves in hybrid models. Additionally,
a numerical comparison between the wave speeds of the
continuum models and the hybrid models shows good agreement
in the case of weak chemotaxis and qualitative agreement for
the strong chemotaxis case. In the case of slow cell adaptation,
we detect oscillating behaviour of the wave, which cannot be
explained by mean-field approximations.

\keywords{hybrid model \and travelling wave \and bacterial chemotaxis}
\end{abstract}

\section{Introduction}

The wavelike spread of cell populations plays a fundamental
role in many biological processes, including development 
\cite{Landman:2007:MEI},
wound healing \cite{Witte:1997:GPW} and tumour invasion
\cite{Gerisch:MMC:2010}. Bacterial populations show similar phenomena,
with the pioneering studies of Adler \cite{Adler:1966:CB} confirming 
the capacity
of an {\em {E. coli}} population to form travelling bands via chemotaxis
to extracellular signals. Mathematically, the extent to which chemotaxis can
generate and sustain {\em {stationary}} travelling bands has motivated 
a number of
studies, including the Keller-Segel model of Adler's experiments 
which is written in the form of coupled partial differential
equations (PDEs) \cite{Keller:1971:TBC}. This early model necessitated 
a biologically
unrealistic singularity in the chemotactic sensitivity to generate 
stationary travelling
waves: a requirement that allows bacteria behind the wave to
acquire infinite speeds and to avoid ``dropping-out'', an effect that 
leads to gradual dispersal of
the band \cite{Xue:2011:TWH,Franz:2012:HMI}.

This singularity requirement can be
circumvented by incorporating other processes. The well known Fisher's
equation \cite{Fisher:1937:WAA} demonstrates travelling waves in systems 
coupling
diffusion with logistic growth terms \cite{Fisher:1937:WAA}. Parabolic
chemotaxis models with non-singular sensitivities but incorporating 
either logistic \cite{Landman:2003:CCM,Landman:2005:DCD,Nadin:2008:TWK} 
or non-logistic \cite{Kennedy:1980:TWS,Satnoianu:2001:TWN} growth terms 
also admit travelling wave solutions. Other studies have
shown that introduction of more complex nutrient terms can give rise
to travelling waves, even when growth is absent 
\cite{Saragosti:2010:MDB,Saragosti:2011:DPC}.
An experimental system which also included two chemicals -- 
a chemo\-attractant and
a nutrient source -- was presented in 
\cite{Budrene:1991:CPF,Budrene:1995:DFS}, with
stationary or transient travelling waves obtained according to t
he formulation of the model \cite{Brenner:1998:PMC,Xue:2011:TWH}. 
Travelling waves in chemotaxis models have also
been recently studied in \cite{Li:2012:SPW,Li:2011:ANS}; we also 
note the articles
\cite{Horstmann:2004:UPK} and \cite{Wang:2013:MTW} for a review 
and analysis of travelling waves in PDE-based models.
A comparison between mesoscopic (hyperbolic) and macroscopic 
(parabolic) PDEs has been
presented in \cite{Lui:2010:TWS}.

Relatively little exploration has been conducted into travelling 
wave formation for chemotactic
models extending beyond PDE systems, in particular those introducing 
terms to account for
the inherent noise of biological systems. One exception is the study 
of \cite{Chavanis:2010:SKS},
in which a multiplicative noise term was introduced into the 
Keller-Segel model and the
existence of travelling waves has been demonstrated within this setting.
Hybrid models, in which an individual-based model for bacterial
behaviour is coupled to a continuum description of extracellular 
signals, naturally
introduce stochastic effects and will be the focus of the present paper.
Such a hybrid model was formulated in \cite{Franz:2012:HMI} where 
it was shown that
under finite cell speeds only transient travelling waves formed, even with
singular chemotactic sensitivity. The individual-based model was 
formulated in terms of the velocity-jump
model with internal dynamics \cite{Erban:2004:ICB,Erban:2005:STS,Xue:2009:MMT} 
and, in this paper,
we extend the model in \cite{Franz:2012:HMI} to incorporate proliferation 
and death
of bacteria. We analyse this system numerically and analytically with respect 
to its
travelling wave properties, employing the biologically inspired chemotactic 
sensitivity
presented in \cite{Xue:2011:TWH} and a linear growth term. We show that 
stationary
travelling waves can be observed even in the absence of chemotaxis, 
although wave
speeds are substantially increased in its presence.

The organisation of the paper is as follows: the full hybrid model 
is presented in Section~\ref{sec:hybrid} along with illustrative
simulation results, while the corresponding continuum equations 
are derived under certain
assumptions in Section~\ref{sec:continuum}; in Section~\ref{sec:travelling} 
these
continuum equations are analysed with respect to travelling wave properties; 
in Section ~\ref{sec:numerics}
% results are compared to the hybrid simulations in Section~\ref{sec:numerics} 
where
a computational analysis and comparison of the models is presented; finally, 
we discuss our
observations in Section~\ref{sec:discussion}.

\section{Hybrid model of bacterial chemotaxis}
\label{sec:hybrid}
In this section we formulate the hybrid model of bacterial chemotaxis 
which will be investigated in
this paper. The model is motivated by the behaviour of the bacterium 
{\it E.coli} and, in its most
general form, includes cell movement, sensing and response to a chemical
signal, consumption of the chemoattractant, cell proliferation and death. 
However,
for analytical tractability, we will also explore simplified hybrid models 
which exclude some of
these processes. Bacteria are modelled as agents with internal dynamics 
that represent
the signal processing and response of each individual while the 
extracellular chemical is modelled
using a PDE to describe its spatio-temporal concentration. The 
mathematical framework and
simulation techniques are reviewed in \cite{Franz:2012:HMI}. We 
consider the model in
an effectively one-dimensional domain representing a long but 
narrow tube, similar to the
experimental set up considered in \cite{Adler:1966:CB}.

The motion of \emph{E. coli} bacteria is controlled through the 
coordinated rotation of
flagella distributed over the cell surface \cite{Berg:1975:HBS}. 
Counterclockwise
rotation generates a propulsive bundle that results in straight 
line motion of the bacterium
-- a so-called ``run'' \cite{Berg:1972:CEC}. Alternatively, clockwise
rotation results in the outward flaying of flagella and a ``tumble'' -- 
rotation with insignificant
displacement. At the end of each tumble the bacterium chooses a new 
direction of
movement, seemingly at random, and returns to the run phase. The 
lengths of
the individual phases are independent from each other and distributed 
exponentially, yet they
can be influenced by internal dynamics \cite{Berg:1975:HBS}.

Internal dynamics of the \emph{E. coli} bacteria possess two principal 
features
\cite{Bourret:1991:STP}: a quick excitation phase followed by slower 
adaptation.
Specifically, changes in the extracellular signal concentration lead 
to quick excitation of
the internal metabolism, signified through altered chemical concentrations 
inside the cell.
Following excitation the internal concentrations revert slowly to normal 
in an adaptation
process, even when the external signal remains at the raised level.

\subsection{Velocity jump model with internal dynamics}
\label{secveljump}
Run-and-tumble dynamics are aptly modelled as a velocity-jump process 
\cite{Othmer:1988:MDB,Erban:2004:ICB}. We denote by $N_a(t)$ the number of
bacteria (agents) in the system at time $t$. The current state of the 
$i$-th individual,
$i=1,2,\ldots,N_a(t),$ will be described using its position $x_i\in\R$, its
velocity $v_i=\pm s \in \R$ and a set of internal state variables
$\mathbf{y}_i \in \R^m$ that represent the states of components in the
intracellular signal transduction network.

Here we concentrate on a cartoon version of the internal
dynamics of bacteria written in terms of two internal variables
\cite{Othmer:1998:OCS,Erban:2004:ICB}, i.e $m=2$. Internal variables
$y^{(1)}$ and $y^{(2)}$ are governed by the equations
\begin{equation}
\begin{aligned}
	\frac{\Der y^{(1)}}{\Der t} &
	= \frac{S(x(t),t) - y^{(1)} - y^{(2)}}{t_e} \,,\\
	\frac{\Der y^{(2)}}{\Der t} &
	= \frac{S(x(t),t) - y^{(2)}}{t_a} \,,
\end{aligned}
\label{eq:internalDyn}
\end{equation}
where $t_e$ is the excitation time, $t_a$ is the adaptation time,
$t_e \ll t_a$ and $S(x(t), t)$ is the concentration of
chemoattractant at the position of the bacterium $x(t)$ at time $t$.
Furthermore, bacteria move with the velocity $v_i = \pm s$ governed
through a velocity jump process with a turning frequency
$\lambda = \lambda(\mathbf{y})$
that depends on the internal dynamics. In this paper, we will use the
biologically motivated nonlinear turning kernel developed in
\cite{Xue:2011:TWH}. Hence, the full model of one individual
over (a small) time step $\Delta t$ can be written as:
\begin{eqnarray}
\label{eq:hybridmodel:x}
&& x(t+\Delta t)  =  x(t) + v(t)\,\Delta t, \\
\label{eq:hybridmodel:v}
&& v(t+\Delta t)  =
\left\{\begin{array}{rl} -v(t), &\mbox{with probability } \,
\lambda(\mathbf{y}(t))\,\Delta t,
\\ v(t), &\mbox{otherwise}\,,
\end{array}\right.\\
\label{eq:hybridmodel:lambda}
&& \lambda(\mathbf{y}(t))  = \lambda_0 \left(1
- \frac{y^{(1)}(t)}{\kappa + |y^{(1)}(t)|}\right)\,,
\label{eq:lamhybrid}\\
\label{eq:hybridmodel:y1}
&& y^{(1)}(t+\Delta t)  =
y^{(1)}(t) + \frac{S(x(t), t) - y^{(1)}(t) - y^{(2)}(t)}{t_e} \,\Delta t,
\\
\label{eq:hybridmodel:y2}
&& y^{(2)}(t+\Delta t)  =
y^{(2)}(t) + \frac{S(x(t), t) - y^{(2)}(t)}{t_a} \,\Delta t,
\end{eqnarray}
where $\lambda_0$ and $\kappa$ are positive constants.

In addition to the behaviour of an individual bacterium we define a 
signal-dependent
proliferation function $h(S): \R^+ \mapsto \R$. We thereby interpret a
positive value of $h(S)$ as a proliferation rate, meaning that in the 
infinitesimal
interval $[t, t+\Delta t)$ a bacterium at position $x$ generates an exact
copy of itself with probability $h(S(x(t), t)) \, \Delta t$. Similarly, 
a negative value
of $h(S)$ means that the bacterium disappears (dies) with the probability
$- h(S(x(t), t)) \, \Delta t$. In this paper, we will use the following 
form for the
proliferation rate $h(S)$:
\begin{equation} \label{eq:defhS}
	h(S) = \alpha (S - S_c)\,,
\end{equation}
where $\alpha$ and $S_c$ are positive constants.

\subsection{Evolution of the extracellular chemoattractant}
\label{secPDES}

For the extracellular signal $S(x, t)$ we formulate a PDE that incorporates
diffusion (with diffusion constant $D_S \ge 0$) and signal consumption by
bacteria, the latter with signal dependent rate $k(S):\R^+\to\R^+$. The 
equation
for $S$ therefore takes the form
\begin{equation} \label{eq:Shybrid}
	\frac{\partial S}{\partial t}
	=
	D_S\frac{\partial^2 S}{\partial x^2}
	-
	k(S)\sum_{i=1}^{N_a(t)} \delta(x-x_i(t))\,.
\end{equation}
For the remainder of the paper we employ a linear form for the consumption
function $k(S)$:
\begin{equation} \label{eq:defnk}
	k(S) = \beta S\,,
\end{equation}
where $\beta$ is a positive constant.

\subsection{Illustrative example}
\label{subsec:hybrid:num}

The hybrid model framework presented in Sections \ref{secveljump} and
\ref{secPDES} includes essential features of the more complicated hybrid
chemotaxis models formulated in \cite{Erban:2005:ICB,Xue:2011:RSS}.
In this section we numerically show that these processes can give rise 
to travelling
waves. For the numerical simulation we employ techniques described
in \cite{Franz:2012:HMI}. In particular, for the extracellular signal 
$S(x, t)$, this means
that the simulation is performed on the one-dimensional domain $[0, L]$ 
with initial
condition $S(x, 0) = S_\infty>0$ and zero-flux boundary conditions. We 
consider
$M+1$ regularly spaced grid points $r_j = j\, \Delta x$, $j=0,\ldots, M,$ 
where
$\Delta x = L/M$ and the values of $S(x_i, t)$ are advanced by a small 
time step
$\Delta t$ and a forward Euler update rule:
\begin{equation} \label{eq:SforwardEuler}
\begin{aligned}
S(r_j, t + \Delta t) & = S(r_j, t) + D_S\, \Delta t
\frac{S(r_{j-1}, t) + S(r_{j+1}, t) - 2S(r_j, t)}{\left(\Delta x\right)^2}\\
& - k(S(r_j, t))\,\Delta t\sum_{i=1}^{N_a(t)} K(r_j - x_i(t))\,.
\end{aligned}
\end{equation}
In the above $K: \R \to \R^+$ is the symmetric, normalised and non-negative 
kernel
\begin{equation*}
	K(\xi) = \frac{1}{\sqrt{2\pi \sigma^2}}
	\exp\left[-\frac{\xi^2}{2\sigma^2}\right]\,,
\end{equation*}
where the kernel width $\sigma$ is a positive real number.
Here, $K(r_j - x_i)$ represents the influence a bacterium
at position $x_i$ has on grid point $j$.

The simulation of the individual bacterium is given in the full system
(\ref{eq:hybridmodel:x})--(\ref{eq:hybridmodel:y2}) and complemented by the
birth and death processes described in Section~\ref{secveljump},
where we use the same time step $\Delta t$ as
in \eqref{eq:SforwardEuler}. To calculate the necessary off-grid
values of extracellular signal, we linearly interpolate from
the two nearest grid points. We further simplify the system
(\ref{eq:hybridmodel:x})--(\ref{eq:hybridmodel:y2}) by exploiting the 
separate
time scales for excitation and adaptation (i.e. $t_e \ll t_a$): specifically, 
we assume
the update equation (\ref{eq:hybridmodel:y1}) for $y^{(1)}$ is in a 
quasi-equilibrium,
which is identical to the assumption $t_e = 0$. The value
for $y^{(1)}$ can therefore be calculated by
\begin{equation}
y^{(1)}(t) = S(x(t),t) - y^{(2)}(t)\,.
\label{reducedy}
\end{equation}
Illustrative results are presented in Figure \ref{fig:hybridSim}. For 
this simulation,
$N_a(0)=10^4$ bacteria were initialised at positions $x_i(0)$, randomly 
generated as
the absolute value of a Gaussian random variable with variance much 
smaller than the
domain length $L$. The initial velocity (direction of movement) is 
generated uniformly at
random and initial values of the extracellular signal and internal 
variables are taken as
\begin{eqnarray*}
y^{(1)}_i(0) = S_\infty, \; \; y^{(2)}_i(0) =  0,
\quad && \mbox{for} \; i=1,2,\dots,N_a(0), \\
S(x,0) \equiv S_\infty
\quad && \mbox{for} \; x \in [0,L],
\end{eqnarray*}
where $S_\infty=1$. We simulate the system until 
time $T_{\mathrm{final}} = 100$ and
plot both the distribution of bacteria and concentration of chemoattractant 
$S$ in
Figure \ref{fig:hybridSim}(a). We also estimate the wave speed as a 
function of time in
Figure \ref{fig:hybridSim}(b).
\begin{figure}[t]
\centerline{\raise 4.3cm
\hbox{(a)}
\hskip -4mm
\includegraphics[height=4.2cm]{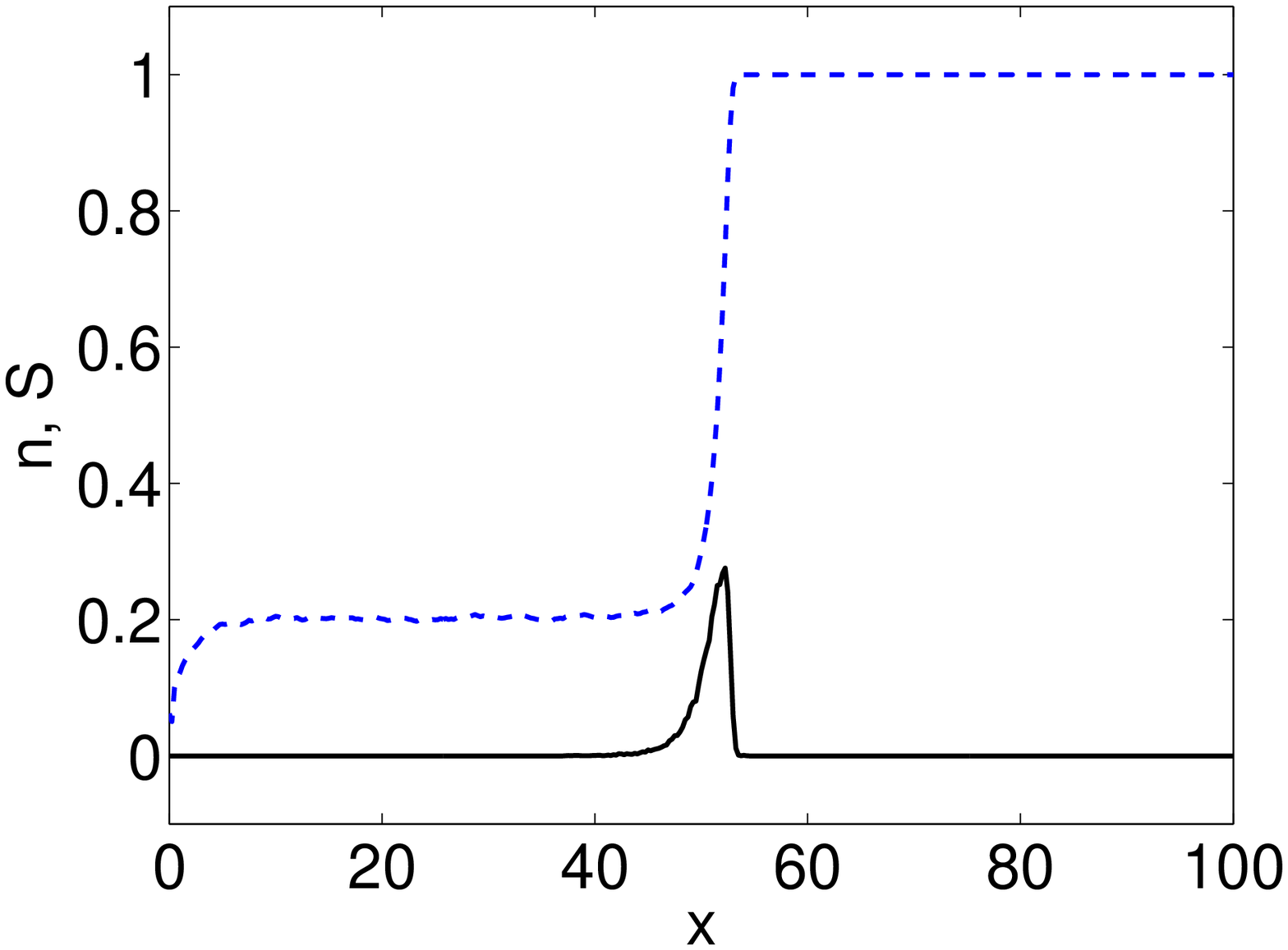}
\hskip 4mm
\raise 4.3cm \hbox{(b)}
\hskip -4mm
\includegraphics[height=4.2cm]{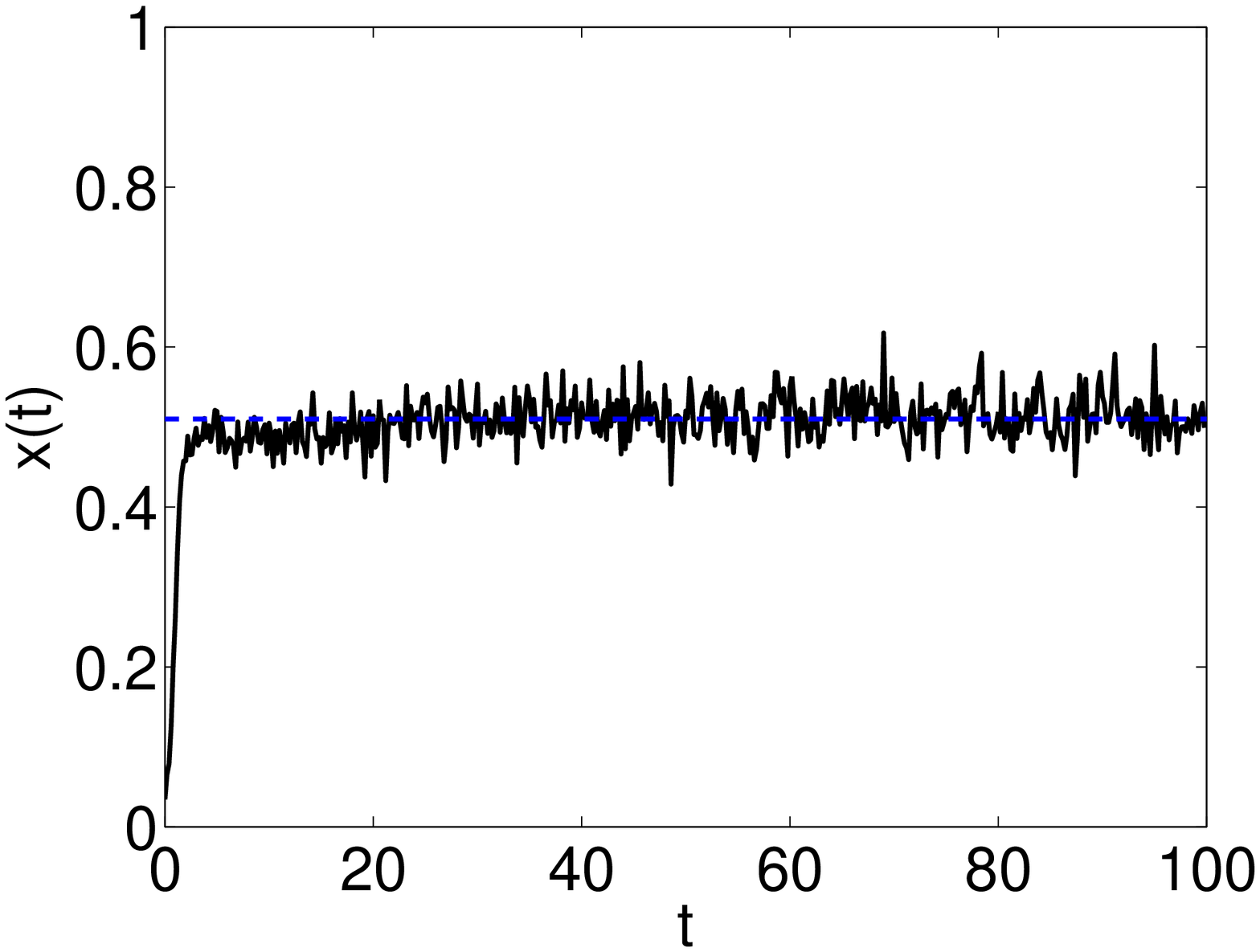}
}
\centerline{\raise 4.3cm
\hbox{(c)}
\hskip -4mm
\includegraphics[height=4.2cm]{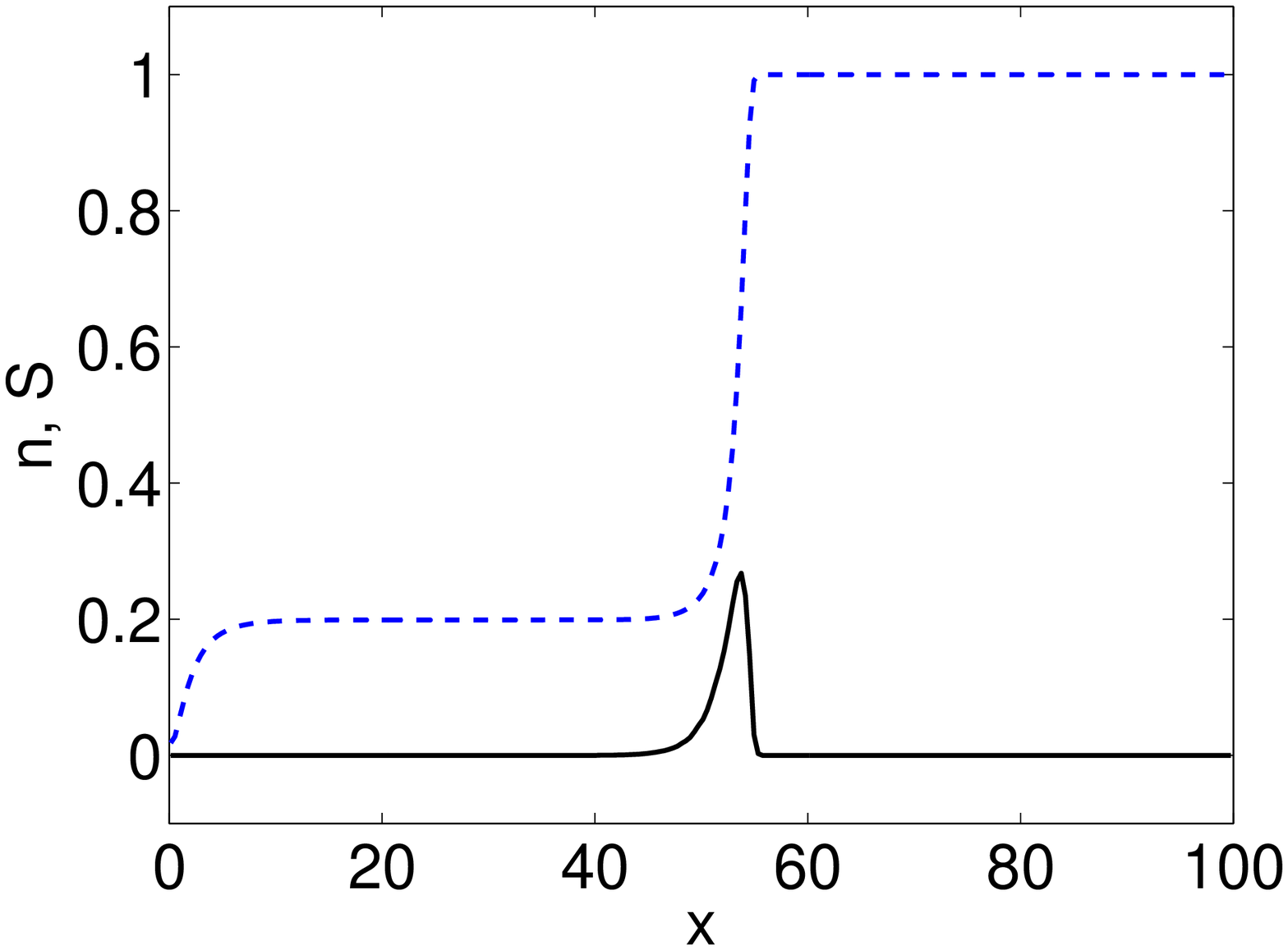}
\hskip 4mm
\raise 4.3cm \hbox{(d)}
\hskip -4mm
\includegraphics[height=4.2cm]{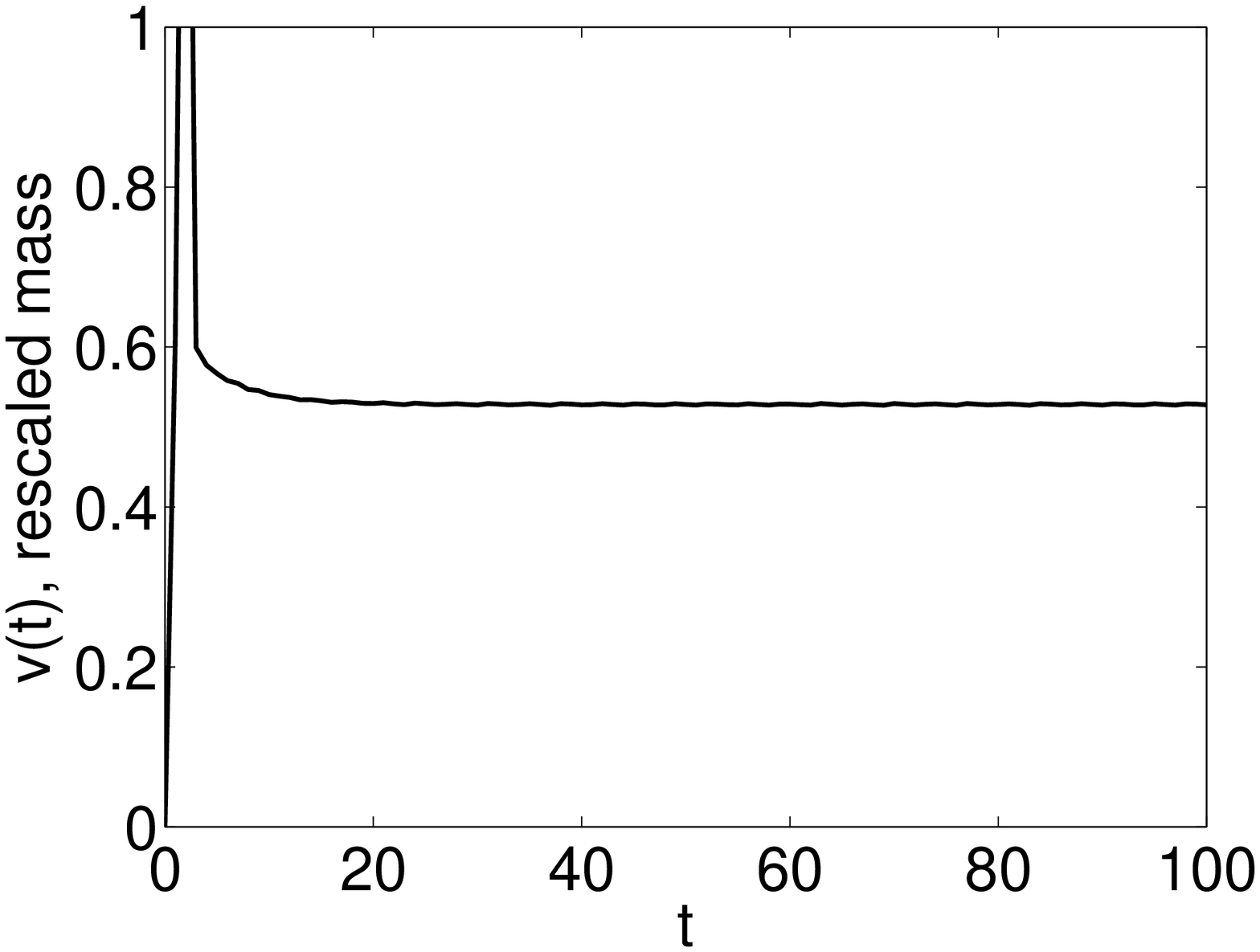}
}
\caption{{\it Numerical solutions of the hybrid chemotaxis model
$(\ref{eq:hybridmodel:x})$--$(\ref{eq:hybridmodel:y2})$
and $(\ref{eq:SforwardEuler})$ and PDE System A  
$(\ref{eq:mesoP})$--$(\ref{eq:mesoS})$.} \hfill\break
(a) {\it Wave form for the hybrid model after time $t = 100$.
Solid line: estimated density of bacteria, dashed line:
extracellular chemical signal $S$.} \hfill\break
(b) {\it Measured speed of travelling wave (solid line).
Dashed line denotes the average speed.}\hfill\break
(c) {\it Wave form for PDE system A after time $t = 100$.
Solid line: estimated density of bacteria, dashed line:
extracellular chemical signal $S$.} \hfill\break
(d) {\it Measured speed of travelling wave (solid line) for PDE System A.
Note that the spike near $t=0$ is a product of the wave speed calculation 
method.
\hfill\break The dimensionless
parameters are: $\alpha = \beta = s = 1$, $S_c = 0.5$, $S_\infty = 1$,
$\Delta t = 10^{-3}$, $\Delta x = 0.25$, $L = 100$, $\lambda_0 = 10$,
$\kappa = 0.01$, $D_S = 0$, $t_a = 0.1$, $\sigma = 0.5$.}
}
\label{fig:hybridSim}
\end{figure}%
We clearly see formation of a travelling band of bacteria, moving 
rightwards with
average speed $v=0.51$ (plotted as the dashed line in Figure 
\ref{fig:hybridSim}(b)).

\subsubsection*{Influence of the growth term}
\label{subsec:comp:nogrowth}
To investigate the influence of the growth term on the existence of
travelling waves, we simulate the full hybrid
model (\ref{eq:hybridmodel:x})--(\ref{eq:hybridmodel:y2}) and 
(\ref{eq:SforwardEuler})
including ($\alpha = 1$) and excluding ($\alpha = 0$) growth and 
death processes.
We use identical parameters to those described above and present 
the results in
Figure~\ref{fig:nogrowth}. In Figure~\ref{fig:nogrowth}(a) the 
position of the wave front
(defined as the right-most position for which $S(x) < 0.9$) is 
compared. The full hybrid
system (dashed line) generates a straight line, indicating a wave 
moving with constant
speed. While the system excluding growth and death (solid line) 
moves with a similar initial speed, speed
is gradually lost over time: the shape of $n(x, t)$ at different 
times for this case is shown in Figure~\ref{fig:nogrowth}(b). We 
clearly see that no true travelling wave forms, with many agents
being left far behind the wave front, leading to its slowing down. 
Thus, we can interpret growth
and death terms in terms of a stabilising role on the wave profile: 
although not all agents can keep
up with the wave, new agents are constantly created at the front and 
the agents that drop out
eventually die, resulting in a travelling band of agents.

\begin{figure}[t]
\centerline{\raise 4.3cm
\hbox{(a)}
\hskip -4mm
\includegraphics[height=4.2cm]{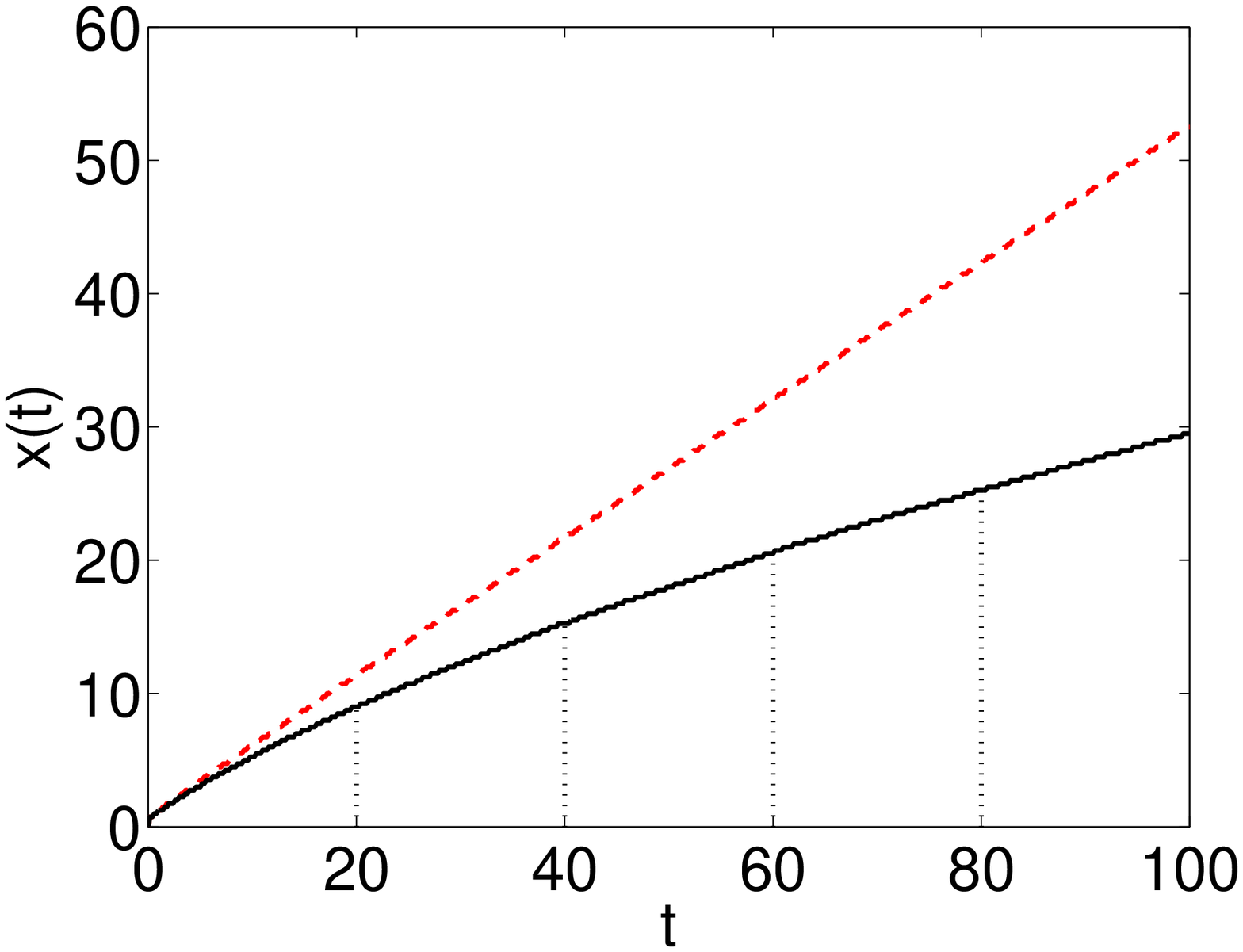}
\hskip 4mm
\raise 4.3cm \hbox{(b)}
\hskip -4mm
\begin{overpic}[height=4.2cm]{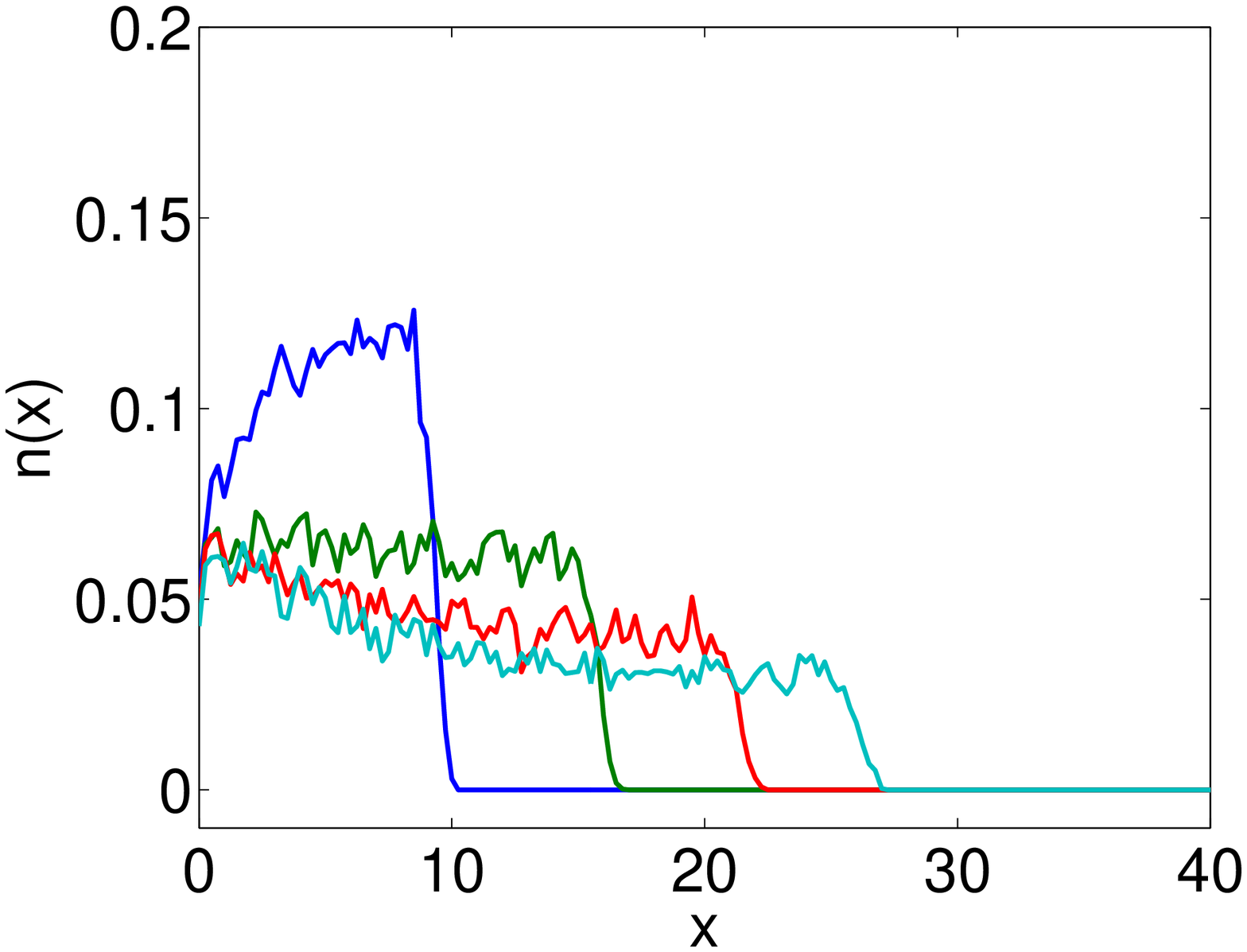}
	\put(34,51){\line(1,1){10}}
	\put(46,61){$t_1 = 20$}
	\put(46,33){\line(1,2){6}}
	\put(53,45){$t_2 = 40$}
	\put(56,29){\line(1,1){8}}
	\put(65,36){$t_3 = 60$}
	\put(65,24){\line(2,1){8}}
	\put(74,27){$t_4 = 80$}	
\end{overpic}
}
\caption{{\it Numerical solutions of the hybrid chemotaxis model
$(\ref{eq:hybridmodel:x})$--$(\ref{eq:hybridmodel:y2})$
and $(\ref{eq:SforwardEuler})$ without growth and death terms.} \hfill\break
(a) {\it Comparison of position of wave front over time.
Solid line: without growth/death ($\alpha = 0$),
dashed line: with growth/death ($\alpha = 1$).}\hfill\break
(b) {\it Wave form at different times during simulation with $\alpha = 0$.
From left to right: $t = 20, 40, 60, 80$.}\hfill\break
{\it Remaining parameters as in Figure~\ref{fig:hybridSim}.}}
\label{fig:nogrowth}
\end{figure}

\section{From hybrid models to macroscopic PDEs}
\label{sec:continuum}
In this section we derive macroscopic PDEs for the spatio-temporal
density of bacteria $n(x, t)$ at given position $x\in\R$
and time $t\geq 0$. An implicit assumption of the derivation is
spatial independence of bacteria, which allows formulation of a 
continuous mesoscopic
system. We then use results from \cite{Erban:2004:ICB} to obtain 
the macroscopic
equations. To illustrate the successive formulation of models we construct
two systems of PDEs -- denoted System~\textbf{(A)}
and System~\textbf{(B)} -- to be referred to in the remainder of the paper.

\subsection{System~\textbf{(A)}}

We define the mesoscopic densities $p^{\pm}(x, y^{(2)}, t)$ for
left and right-moving bacteria, depending on their position
$x\in\R$, their internal variable $y^{(2)}\in\R$ and
$t \geq 0$. If the signal profile $S \equiv S(x,t)$ was uninfluenced 
by bacteria,
densities $p^{\pm}$ would satisfy the following system of hyperbolic 
PDEs:
\begin{equation}
\begin{aligned}
&	\frac{\partial \pps}{\partial t} + s\frac{\partial \pps}{\partial x}
	+ \frac{\partial}{\partial y^{(2)}} \left(\frac{S(x,t)-y^{(2)}}{t_a}\pps
	\right)& = & -\lambda \pps + \lambda\pms + h(S(x,t))\pps\,,\\
&	\frac{\partial \pms}{\partial t} - s\frac{\partial \pms}{\partial x}
	+ \frac{\partial}{\partial y^{(2)}} \left(\frac{S(x,t)-y^{(2)}}{t_a}\pms
	\right)& = & {\color{white}+}\lambda \pps - \lambda\pms + h(S(x,t))\pms\,,
\end{aligned}
\label{eq:mesoP}
\end{equation}
where $\lambda$ is defined in \eqref{eq:lamhybrid} which,
under (\ref{reducedy}), can be simplified to
\begin{equation}
\lambda = \lambda_0
\left(1 - \frac{S(x)-y^{(2)}}{\kappa + |S(x)-y^{(2)}|}\right)\,.
\label{eq:lamhybridred}
\end{equation}
The signal dynamics is described by \eqref{eq:Shybrid} which can be 
rewritten in terms of
$p^{\pm}$ as \begin{equation}
\frac{\partial S}{\partial t} =
D_S \frac{\partial^2 S}{\partial x^2}
- k(S) \int_{\R} (\pps + \pms) \mathrm{d} y^{(2)}\,.
\label{eq:mesoS}
\end{equation}
We denote the system of equations \eqref{eq:mesoP}--\eqref{eq:mesoS}
as System~\textbf{(A)}.

The system (\ref{eq:mesoP}) (for the one-particle distribution)
can be derived by integrating the probability distribution 
function $p(x_1,v_1,{\mathbf y}_1;x_2,v_2,{\mathbf y}_2;\dots \,|\, S(x,t))$ 
for the many particle system, utilizing the fact that the movement 
of individuals are biased by the signal function $S(x,t)$, but 
independent to each other. However, for the hybrid chemotaxis 
models described in Sections~\ref{secveljump} and \ref{secPDES}, 
individual bacteria interact via the extracellular signal $S$
which complicates the derivation of (\ref{eq:mesoP}).
In \cite{Erban:2012:ICB}, a kinetic description has been derived for a model 
of interacting locusts, using a modified version of the BBGKY hierarchy 
from the classical
kinetic theory of gases \cite{Cercignani:1994:MTD}. The system we consider
 here is much more
complicated to analyse than the locust model studied in \cite{Erban:2012:ICB},
due to the variable number of bacteria and internal variables.
Thus the kinetic description (\ref{eq:mesoP}) can only 
be considered as an approximation to the one particle distributions 
of the interacting system. 

The capacity of the above
mesoscopic system to generate travelling bands analogous to those
observed in the hybrid model is illustrated in Figure \ref{fig:hybridSim}(c)-(d).
For details of the numerical method employed for this and other simulations of
the continuous model, we refer to \cite{Xue:2011:TWH}.
The qualitatively and quantitatively close correspondence in solutions under
equivalent parameters and initial conditions  corroborates the
use of the above approximation.

\subsection{System~\textbf{(B)}}
We consider a macroscopic model in this section. Define the macroscopic densities
\begin{equation}
\label{defppm}
p^{\pm}(x, t) = \int_{\R} p^{\pm}(x, y^{(2)}, t) \mathrm{d}y^{(2)},
\end{equation}
and let them satisfy the following system
\begin{equation}
\begin{aligned}
& \frac{\partial \pps}{\partial t} + s \frac{\partial \pps}{\partial x}
=
- \lambda^{+}\left(\gradS \right) \pps + \lambda^{-}\left(\gradS\right)
\pms + h(S) \pps\, ,
\\
& \frac{\partial \pms}{\partial t} - s \frac{\partial \pms}{\partial x}  =
{\color{white}+} \lambda^{+}
\left(\gradS \right) \pps - \lambda^{-}\left(\gradS \right) \pms +
h(S) \pms\,,
\end{aligned}\label{eq:macp}
\end{equation}
where
the turning rates  $\lambda^{\pm}$ are given by
\begin{equation} \label{eq:lamcont}
\lambda^{\pm}
=
\lambda_0
\left(1 \mp \chi \gradS \right)  \quad \mbox{with} \quad
\chi = \frac{s t_a}{\kappa \lambda_0 (1+2\lambda_0 t_a)}\,.
\end{equation}
Using (\ref{defppm}), equation (\ref{eq:mesoS}) can be written
as
\begin{equation}
\label{Seqbef}
\frac{\partial S}{\partial t} =
D_S \frac{\partial^2 S}{\partial x^2} - k(S) (\pps + \pms) \,.
\end{equation}
We will denote \eqref{eq:macp} and \eqref{Seqbef} along with the definition
of $\lambda^\pm$ in \eqref{eq:lamcont} as System~\textbf{(B)}.
According to the analysis in \cite{Erban:2004:ICB,Xue:2009:MMT},
System~\textbf{(B)} is quantitatively consistent with System~\textbf{(A)} when the external
signal $S(x)$ changes slow enough such that cells are close to their fully adapted state, in
which case cell movement is only moderately modified by the signal.

\bigskip
In the rest of the paper, we assume diffusion of extracellular
signal to occur on a much slower time scale than the
active motion of the bacteria, hence $D_S = 0$. The number of parameters of the
above models can be reduced by setting $s, S_{\infty}, \alpha, \beta$ to 
one through
rescaling. We show this in detail for System~\textbf{(B)} as follows.
Rescaling the variables
$S =  \hat{S} S_\infty$, $p^{\pm} = \hat{p}^\pm \alpha S_\infty/\beta$,
$t = \hat{t}/(\alpha S_\infty)$, $x = \hat{x}s/(\alpha S_\infty)$
and the parameters $S_c = \hat{S}_c S_\infty$,
$\lambda_0 = \hat{\lambda}_0 \alpha S_\infty$, taking (\ref{eq:defhS}) and
substituting into System~\textbf{(B)} we obtain, after dropping 
hats for notational simplicity,
\begin{equation}
\begin{aligned}
\frac{\partial \pps}{\partial t} +   \frac{\partial \pps}{\partial x}
& =   - \lambda^{+}\left(\gradS \right) \pps + \lambda^{-}\left(\gradS\right)
\pms + (S-S_c) \pps \,, \\
\frac{\partial \pms}{\partial t} -   \frac{\partial \pms}{\partial x}
& =   {\color{white}+} \lambda^{+}\left(\gradS \right) \pps -
\lambda^{-}\left(\gradS \right) \pms +  (S-S_c) \pms \,,  \\
\frac{\partial S}{\partial t} & =    - S (\pps + \pms) \,.
\end{aligned} \label{sys_pS2}
\end{equation}
We are interested in travelling wave solutions that develop from a pointwise inoculation of cells
into a domain containing uniformly distributed nutrient $S$. In this
scenario, $p^{\pm}$ (defined as in each system) should form travelling pulses while $S$ forms a
travelling front and relevant boundary conditions will be
\begin{equation}
\begin{aligned}
& p^{\pm}, \pd{p^{\pm}}{x}, \pd{S}{x}\to 0  && \mbox{as}\quad
x \to \pm \infty \,,\\
& S \to 1 && \mbox{as}\quad  x \to  + \infty\,, \\
& S \to S_- && \mbox{as} \quad x \to  - \infty\,.
\end{aligned}\label{bc_pS2}
\end{equation}
Note that $S_-$ is currently unknown; we determine its value in
the travelling wave analysis of Section \ref{sec:travelling}.
Since $p^{\pm}$ and $S$ are physical quantities, we search
for nonnegative travelling wave solutions, i.e.
$$
p^{\pm}\geq 0, \quad S \geq 0.
$$
It is clear that a travelling wave of this form cannot exist
for $S_c \geq 1$ (extinction of bacteria) or for $S_c \leq 0$
(infinite growth) and we will therefore only consider systems
that satisfy $S_c \in (0, 1)$. In the next section we analyse
System~\textbf{(B)} with respect to travelling wave solutions
in order to obtain further insight. To do that, we use the rescaled
system (\ref{sys_pS2}).

\section{Travelling wave analysis}
\label{sec:travelling}

In this section we first apply the standard travelling
wave ansatz to system \re{sys_pS2} and derive a necessary condition
for the existence of non-negative travelling wave solutions. We then
reduce the resulting ODE system to two components through a change of 
variables
and utilizing an invariant manifold identified for the problem. Finally
we use phase plane methods to analyse the existence and properties of
travelling wave solutions.

\subsection{A necessary condition for the existence of travelling
wave solutions}
Let us apply the travelling wave ansatz
$p^{\pm}(x, t) = p^{\pm}(\xi) = p^{\pm}(x-ct)$
and $S(x, t) = S(\xi) = S(x-ct)$,
where $c$ is the unknown wave speed \cite{Murray:2002:MB}.
System \re{sys_pS2} becomes
\begin{equation}
\begin{aligned}
 (1-c)(p^+)' & =
 - \lambda_0 \left( 1-\chi \,S'\right) \pps
 +  \lambda_0 \left( 1+\chi\,S'\right) \pms
 + (S-S_c) \pps \,, \\
-(1+c)(p^-)' & =
{\color{white}+}  \lambda_0 \left( 1-\chi\,S'\right) \pps -
\lambda_0 \left( 1+\chi\,S'\right) \pms + (S-S_c) \pms \,,  \\
-c S' & =    - S (\pps + \pms) \,,
\end{aligned}  \label{twsys_pS}
\end{equation}
where the primes denote derivatives with respect to the travelling
wave variable $\xi$. Note that any point on the $S$-axis
is a steady state of the system \re{twsys_pS} and that linear
stability of such a steady state, $(\pps,\pms,S) = (0, 0, S_*)$, is
governed by the eigenvalues of the matrix $A^{-1}B$, where
$$
A =
\begin{pmatrix}
1-c & 0 & 0 \\ 0 & -1-c & 0 \\ 0 & 0 & -c
\end{pmatrix}\, , \qquad
B =
\begin{pmatrix}
\ds-\lambda_0 + S_* - S_c & \lambda_0 & 0 \\
\ds \lambda_0 & -\lambda_0 + S_*  - S_c& 0 \\
- S_* & - S_* & 0
\end{pmatrix}\,.
$$
The eigenvalues of $A^{-1}B$ are
\begin{equation}
\mu_1 = 0, \qquad \mu_{2,3} = \frac{ c(-\lambda_0 + S_* - S_c)
\pm \sqrt{\Delta_1(S_*) }}{1-c^2} \,,
\label{fullsys_eigenvalues}
\end{equation}
where
\begin{equation}
\Delta_1 (S_*)= c^2 \lambda_0^2  +  (S_* - S_c - 2\lambda_0)(S_*- S_c).
\label{simdel}
\end{equation}
Under the boundary conditions \re{bc_pS2} we
look for nonnegative solutions to \re{twsys_pS} connecting steady
states $(\pps, \pms, S) = (0, 0, S_{-})$ and $(\pps,\pms,S) =(0, 0, 1)$. 
To admit such a solution the latter
must be a stable node, since a stable spiral would imply negative values
for $p^{\pm}$. Hence, a necessary condition is $\Delta_1(1) \geq 0$, which
is equivalent to
\begin{equation}
\label{eq:cstar}
c\geq c^* = \frac{1}{\lambda_0} \sqrt{(2\lambda_0 - 1 + S_c) (1-S_c)}\,.
\end{equation}
Given $2\lambda_0 > (1-S_c)$ it is easy to show that
$c^*\in [0,1]$.

\begin{thm}
A necessary condition for the existence of nonnegative travelling 
wave solutions of the system $(\ref{sys_pS2})$ is
\begin{equation}
2\lambda_0 > (1-S_c). \label{prop1}
\end{equation}
\end{thm}
The above condition is reasonable, as we expect the run duration
to occur on a much faster time scale than proliferation processes.

\subsection{Dimension reduction}

Let us now perform a change of variables by introducing the cell
density $n = p^+ + p^-$ and the cell flux
$j = p^+ - p^-$. The travelling wave system \re{twsys_pS} can then
be written as
\begin{eqnarray}
-c n' + j' & = & ( S - S_c )n\,, \label{eqn_nw} \\
-c j' +    n' & = & 2\lambda_0 \chi\, S'n   + (  S   - S_c - 2\lambda_0 ) j\,,
\label{eqn_jw}\\
- c S' & = & - Sn\,, \label{eqn_Sw}
\end{eqnarray}
where the boundary conditions for this system are
$$
\begin{aligned}
& n, j, \pd{n}{x}, \pd{j}{x}, \pd{S}{x}\to 0  && \mbox{as}\quad
\xi \to \pm \infty \,,\\
& S \to 1 && \mbox{as}\quad  \xi \to  + \infty\,, \\
& S \to S_- && \mbox{as} \quad \xi \to  - \infty\,.
\end{aligned}\label{bc_nj}
$$
From   \re{eqn_Sw}, we have $Sn = c S'$ and, hence, $n =  c(\ln S)'$.
Substituting into \re{eqn_nw} we obtain
$$
- c n' + j' =  c S' -   c S_c (\ln S)'\,.
$$
Integrating and applying the boundary conditions at
$\xi \to +\infty$, an invariant manifold of the problem is given
by
$$
- c n + j = c  (S-1) - c S_c  \ln S \,.
$$
With the definition $f(S)\equiv S-1 - S_c  \ln S$, we obtain
$j = cn + c f(S)$, which can be used to eliminate $j$ from the
system \re{eqn_nw}--\re{eqn_Sw}. For $c\neq 1$ we can solve for
$n'$ and obtain the reduced system
\begin{eqnarray}
n' &=&  \frac{c}{1-c^2} \left[ \frac{ 2\lambda_0 \chi\, S n^2}{c^2}
+  2\,n (S - S_c - \lambda_0) + ( S -S_c -2\lambda_0) f(S) \right],
\label{ode1} \qquad \\
S' &=& \frac{1}{c}Sn\,. \label{ode2}
\end{eqnarray}
For $c=1$, we obtain
\begin{eqnarray}
n &=& \frac{\lambda_0 - S + S_c - \sqrt{(\lambda_0 - S + S_c)^2
- 2\lambda_0\chi\,
S (S-S_c - 2\lambda_0)f(S)}}{2\lambda_0\chi\, S} \,, \qquad
\label{ode1_c1} \\
S' &=& Sn\,, \label{ode2_c1}
\end{eqnarray}
where we chose the solution to the quadratic equation for $n$ that
satisfies the boundary conditions $n\to 0$ as $\xi\to\pm\infty$.

It can be easily shown that $f(S) = 0$ has two solutions in the
region $(0, 1]$ for all $S_c \in (0, 1)$ as follows. Since $f'(S) = 1 - S_c/S$,
$f(S)$ is monotonically decreasing for $S\in(0, S_c)$ and monotonically
increasing for $S\in (S_c, 1]$. With $f(1) = 0$, this
implies $f(S_c) < 0$ and, using $f(S)\to \infty$
for $S\to 0$, we obtain the existence and uniqueness of the
second root of $f(S) = 0$: we call it $S_1 \in (0, S_c)$. The
existence of $S_1$ and the negativity of $f(S)$ for $S\in(S_1, 1)$,
together with the condition $2\lambda_0 > 1 - S_c$, implies that $n$
as given in \re{ode1_c1} is positive everywhere, and that the
given solution therefore satisfies the nonnegativity condition.

\subsection{Steady states and their linear stability}

Using the two roots of $f(S) = 0$ and under the condition \eqref{prop1},  
it is clear that there are two
steady states of the system \re{ode1}-\re{ode2}:
$(n,S) = (0, 1)$ and $(n,S) = (0, S_1)$. Linearising the
system  \re{ode1}-\re{ode2} about its steady states
generates a system of the form
$$
\begin{pmatrix}
n\\
S
\end{pmatrix}' =
A
\begin{pmatrix}
n\\
S
\end{pmatrix},
$$
where, for the general steady state $S_* \in \{S_1, 1\}$, we have
$$
A = \begin{pmatrix}
\ds\frac{2c}{1-c^2}\left(S_* - S_c - \lambda_0\right) &
\qquad\ds\frac{c}{1-c^2} (S_* - S_c - 2\lambda_0)\frac{S_* - S_c}{S_*}
\vspace*{0.3cm}\\
\ds\frac{S_*}{c} & 0
\end{pmatrix}
$$
with
$$
\mathop{\mbox{trace}} A = \frac{2c}{1-c^2}\left(S_* - S_c - \lambda_0\right)\,,
\quad
\det A = -\frac{1}{1-c^2} (S_* - S_c - 2\lambda_0)(S_* - S_c)\,.
$$
The eigenvalues of $A$ are identical to $\mu_{2,3}$ as
given in \re{fullsys_eigenvalues}. The steady state
$(0, 1)$ is therefore a stable node for all $c \in (c^*, 1)$ with
 $c^*$ as defined in \re{eq:cstar}. Similarly, it can be seen
that the steady steady $(0, S_1)$ is a saddle point. The eigenvectors
corresponding to the eigenvalues $\mu_{2,3}$ take the form
$$
\ds\bv_{1,2} =
\begin{pmatrix}
\ds \mu_{2,3} \; , & \ds \frac{S_*}{c}
\end{pmatrix}^T \,.
$$
In the $n-S$ plane, the slopes of the eigenvectors are given by
\begin{equation}
k_{1,2} ( c )   =  \frac{\mu_{2,3}c}{S_*}\,.
\end{equation}
For the steady state $(n, S) = (0, 1)$ this slope can be written in the form
\begin{equation}
k_{1,2} ( c ) = \frac{c^2 \lambda_0^2}{1 - S_c - \lambda_0 \mp \sqrt{\Delta}}\,,
\label{eq:k12:alt}
\end{equation}
where we define $\Delta = c^2\lambda_0^2 + (1 - S_c - 2\lambda_0)(1 - S_c)$
similarly to (\ref{simdel}).

\subsection{ Case I: No chemotaxis ($\kappa = \infty$) }
\label{subsec:analysis:nochem}

We first consider the case where the chemotactic sensitivity
$\chi$ (given by (\ref{eq:lamcont})) vanishes, i.e cells do not
respond chemotactically to changes in $S$. Here, travelling waves
are generated solely through proliferation of bacteria at the wave front.
To understand the wave behaviour we perform a phase
plane analysis for the ODE system (\ref{ode1})--(\ref{ode2}).
Using $\kappa=\infty$ (i.e. $\chi=0$), it reduces to
\begin{eqnarray}
\begin{aligned}
 n' &=  \frac{c }{1-c^2} \Big[
 2n(S - S_c - \lambda_0) + ( S -S_c -2\lambda_0)f(S)
 \Big]\,,  \\
 S' &= \frac{1}{c}Sn\,.
\end{aligned}
\label{ode1_S}
\end{eqnarray}
Thus, the slope of a trajectory in the $n-S$ plane can be written as
$$
\dd{n}{S} =  \frac{c^2 }{1-c^2} \frac{2n(S - S_c - \lambda_0) +
( S -S_c -2\lambda_0)f(S) }{Sn}\,.
$$
Additionally, an expression for the $n-$nullcline $\Gamma_n $ is
given by
$$
n = -\frac{S - S_c - 2\lambda_0}{2(S - S_c - \lambda_0)}f(S)\,,
$$
and the $S$-nullcline is simply
$$
n = 0, \quad \mbox{or} \quad S = 0\,.
$$
Let us now show that travelling waves exist for the reduced
system \re{ode1_S}.

\vskip 2mm

\begin{thm}\label{thm2}
For the case $\chi = 0$ (which is equivalent to $\kappa = \infty$),
a unique travelling wave solution for the system
$(\ref{sys_pS2})$ exists for all $c\in(c^*, 1)$.
\end{thm}

\begin{proof}
For any $c \in (c^*, 1)$ we can define a region $\Omega$ 
(see Figure~\ref{fig:nochem}(a)),
enclosed by the line $n = k_2(S - 1)$ (with $k_2$ defined 
in \eqref{eq:k12:alt}),
the $S$-nullcline $n = 0$ and the line $S = S_1$. We will 
first show that $\Omega$
\begin{figure}[t]
\centerline{\raise 4.3cm
\hbox{(a)}
\hskip -4mm
\begin{overpic}[height=4.2cm]{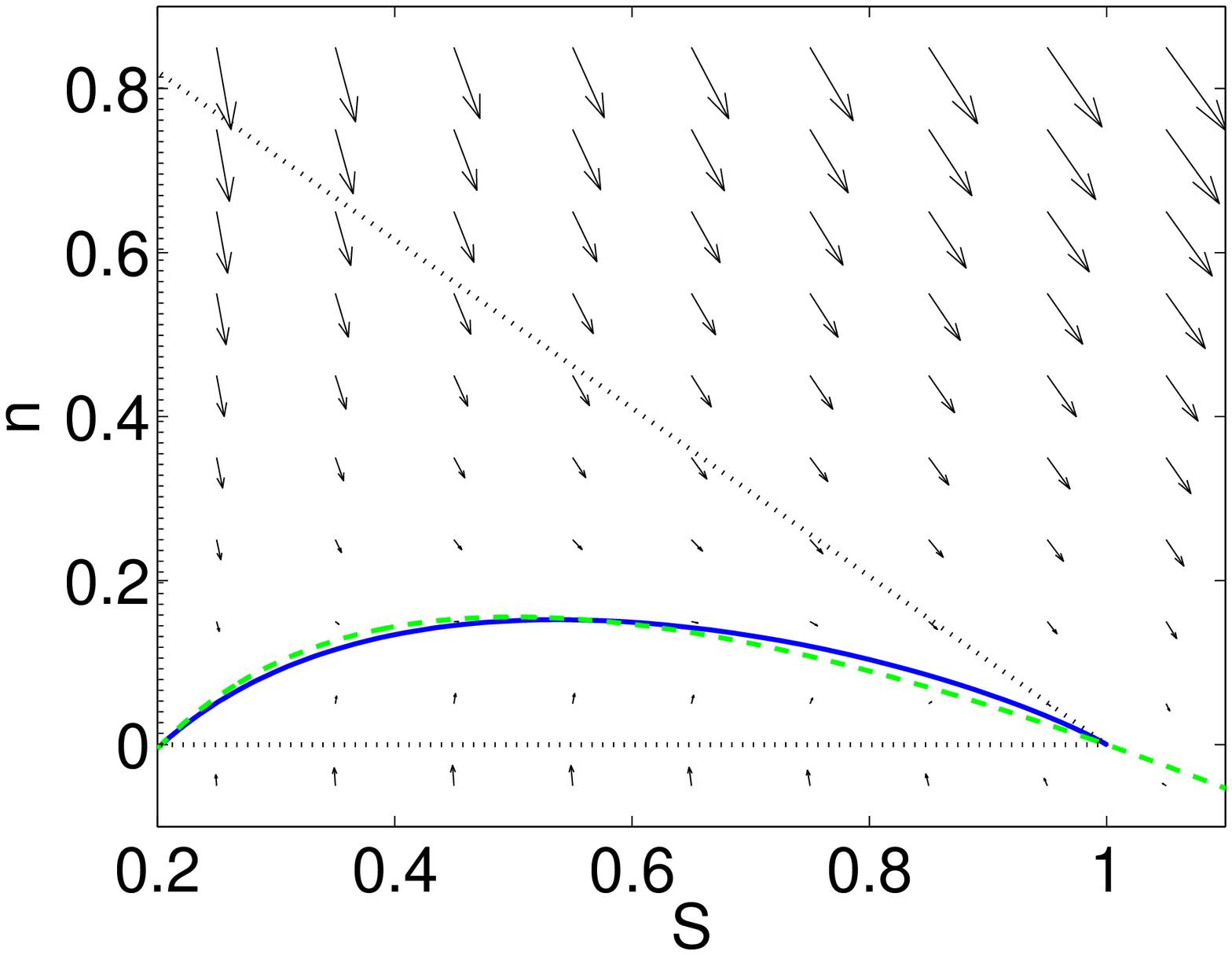}
	\put(30,35){\Large$\Omega$}
\end{overpic}
\hskip 4mm
\raise 4.3cm \hbox{(b)}
\hskip -4mm
\includegraphics[height=4.2cm]{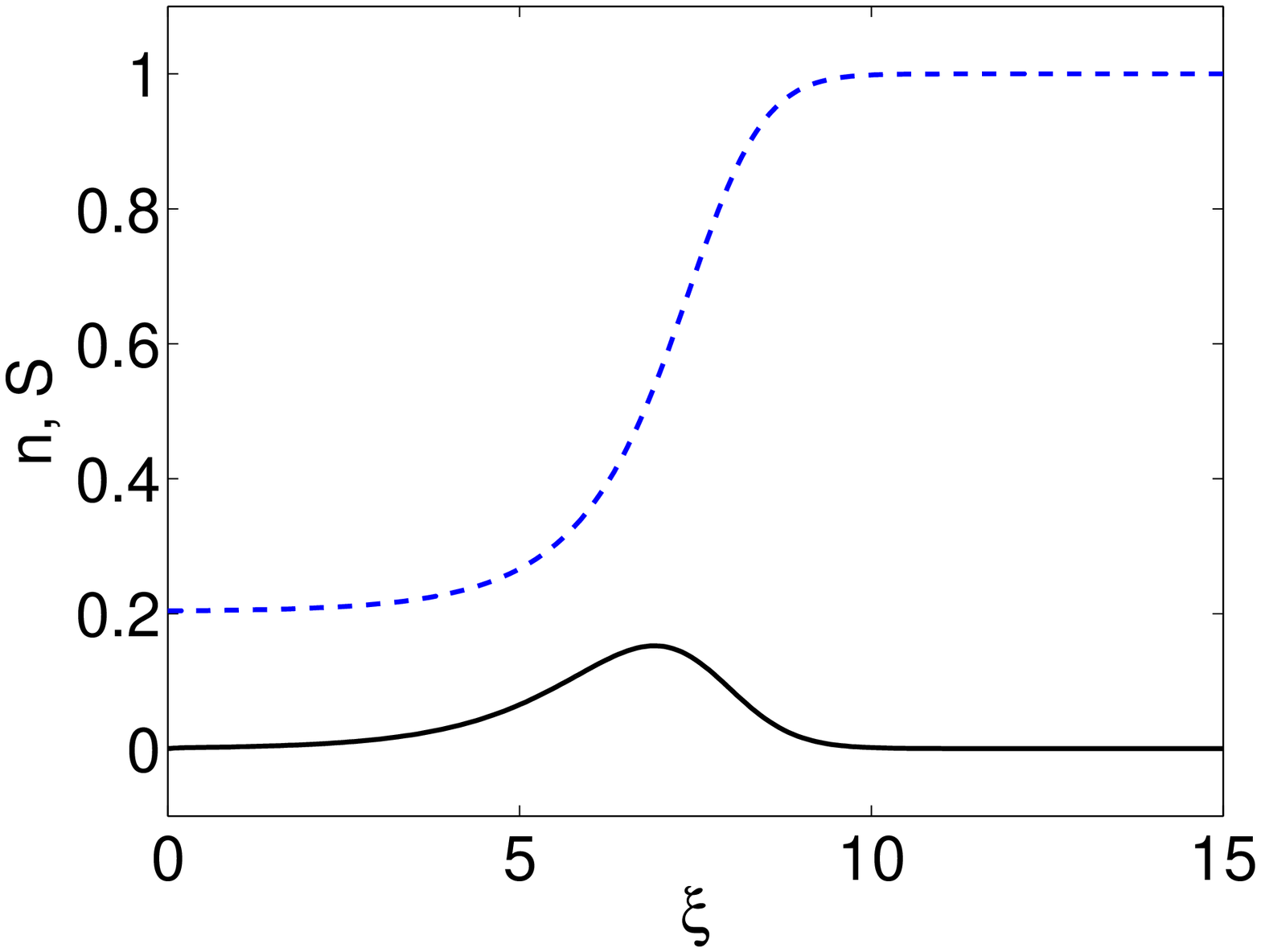}
}
\caption{{\it Illustration of the travelling wave solution
calculated using the ODE system \eqref{ode1}--\eqref{ode2}
for $\chi = 0$,
$\lambda_0 = 10$, $c = c^* = 0.3122$ and $S_c = 0.5$.} \hfill\break
(a) {\it Trajectory of travelling wave solution. Solid line:
trajectory, dashed line: $n$-nullcline, dotted line:
	circumference of invariant region $\Omega$ introduced in the 
	proof of Theorem \ref{thm2}.} \hfill\break
(b) {\it Travelling wave solution in $\xi$.
	Solid line: $n$, dashed line: $S$.\hfill\break
}
}
\label{fig:nochem}
\end{figure}%
is an invariant region of the system \eqref{sys_pS2}.
Since $S$ is non-decreasing everywhere in $\Omega$ and $n'$
is non-negative for $n=0$ and $S\in[S_1, 1]$, we need only to
show that the direction field on the segment
$ \Gamma_1  = \{(n,S): n = k_2(S - 1), S\in [S_1, 1)\}$
points from the top half of the plane above this segment towards
the bottom. Since $S$ is strictly increasing we require
$$
\left.\dd{n}{S}\right|_{\Gamma_1} \leq k_2 \; (\leq 0)\,.
\label{gamma0_direction}
$$
Indeed,
\begin{eqnarray*}
\frac{1-c^2}{c^2} \left.\dd{n}{S}\right|_{\Gamma_1} \!&=&
2 \frac{S-S_c-\lambda_0}{S} + \frac{(S-S_c-2\lambda_0)f(S)}{S(S-1)k_2}\,, \\
& = &  2 \frac{S-S_c-\lambda_0}{S} +
\frac{(S-S_c-2\lambda_0)f(S)}{S(S-1)c^2\lambda_0^2}(1 - S_c - \lambda_0 +
\sqrt{\Delta})\,, \\
&\leq& 2 \frac{S-S_c-\lambda_0}{S} +
\frac{(S-S_c-2\lambda_0)(1-S_c)}{S c^2 \lambda_0^2} (1 - S_c - \lambda_0 +
\sqrt{\Delta})\,, \\
\end{eqnarray*}
where we used \eqref{eq:k12:alt} in the first step and the
relation $f(S)/(S-1) \leq 1 - S_c$ for all $S\in[S_1, 1]$.
Using the fact that $k_2$ and $(S-S_c-2\lambda_0)$ are negative,
we can use the definition of $c^*$ and the fact that
$S\leq 1$ to obtain
\begin{eqnarray*}
\frac{1-c^2}{c^2} \left.\dd{n}{S}\right|_{\Gamma_1}
&\leq& 2 \frac{S-S_c-\lambda_0}{S}
- \frac{2\lambda_0 + S_c - S}{S(2\lambda_0 + S_c - 1)}(1 - S_c - \lambda_0
+ \sqrt{\Delta})
\,,\\
& = & \frac{-2\lambda_0^2 + 3\lambda_0(S-S_c) - (S-S_c)(1-S_c)}{S(2\lambda_0
+ S_c - 1)}
\\ &&- \frac{2\lambda_0 + S_c - S}{S(2\lambda_0 + S_c - 1)}\sqrt{\Delta}\,,\\
& \leq &
\frac{-2\lambda_0^2 + 2\lambda_0(S-S_c) + \lambda_0(1-S_c)- (S-S_c)(1-S_c)}%
{S(2\lambda_0 + S_c - 1)}
\\&& - \frac{2\lambda_0 + S_c - S}{2\lambda_0 + S_c - 1}\sqrt{\Delta}\,,\\
& = & -\frac{(2\lambda_0 + S_c - 1)(\lambda_0 + S_c - S)}{S(2\lambda_0 + S_c
- 1)} -
\frac{2\lambda_0 + S_c - S}{2\lambda_0 + S_c - 1}\sqrt{\Delta}\,,\\
& \leq & -\lambda_0 + 1 - S_c - \sqrt{\Delta} = \frac{1-c^2}{c^2} k_2\,,
\end{eqnarray*}
where we used $S\leq 1$ throughout the derivation. We can therefore
conclude that $\Omega$ is an invariant region
of the system \eqref{sys_pS2}. Noting that at the steady state
$(n, S) = (0, S_1)$ the unstable manifold has a positive
slope ($k_{1,2} = \mu_{2,3}c / S_*$), i.e. it points into
the region $\Omega$, and using the fact that $S$ is
strictly increasing inside $\Omega$ for $n>0$ we can conclude
that, for each $c\geq c^*$, there is a heteroclinic orbit starting
from $(0, S_1)$ and finishing at $(0,1)$, corresponding to a travelling
wave solution of the PDE system \eqref{sys_pS2}. \q
\end{proof}

\subsection{Case II: Increasing chemotaxis ($ 0 < \kappa < \infty$) }

Decreasing $\kappa$ corresponds to an increase in the chemotactic
sensitivity $\chi$ in the ODE system \eqref{ode1}--\eqref{ode2} and
the slope of trajectories in the $n-S$ plane is determined by
$$
\dd{n}{S} =  \frac{c^2 }{1-c^2} \frac{2n(S - S_c - \lambda_0) +
( S -S_c -2\lambda_0)f(S) }{Sn} + \frac{2\lambda_0\chi}{1-c^2} n\,.
$$
It is noted that the above slope is larger than that for the non-chemotaxis
case within the region of interest $n > 0$. Due to this increase the
region $\Omega$ for the proof of Theorem 1 is no longer invariant
for this system and a travelling wave solution to \eqref{sys_pS2} does
not necessarily exist for all $c\in(c^*, 1)$.
The $n$-nullcline for the full ODE system
\eqref{ode1}--\eqref{ode2} is given as the solution of the
quadratic equation
$$
 \frac{ 2\lambda_0\chi S}{c}n^2  +  2c (S-S_c-\lambda_0) n
 + c( S -S_c -2\lambda_0) f(S) = 0.
$$
For a given wave speed $c$, the $n$-nullcline can
therefore be calculated as
\[
n = \frac{c}{2\lambda_0\chi S} \left[c(\lambda_0 + S_c - S)
	\pm \sqrt{\Delta_2(S)} \right]\,,
\]
with	
\[
	\Delta_2(S) = c^2(\lambda_0+S_c-S)^2
	- 2\lambda_0\chi\,S(S-S_c-2\lambda_0)f(S)\,.
\]
We can see that $\Delta_2(S) \to -\infty$ as $S\to\infty$
due to its leading order term $-2\lambda_0\chi\,S^3$.
Therefore, as $S$ becomes large, no $n$-nullcline exists and
$n'$ is positive everywhere. Additionally, $\Delta_2(S)$ might
have further roots and, in particular, $\Delta_2(S)$ might be
negative in parts (or the whole) of region $S\in[S_1, 1]$. This
again means that $n$ is strictly growing in these parts of
the domain.

We detect three different types of behaviours of
trajectories starting close to $(n, S) = (0, S_1)$, plotted
in Figure~\ref{fig:trajectories}. In particular, we can
see each of these behavioural types for different values of
$\chi$ and despite different configurations of the nullclines.
In the top two plots of Figure~\ref{fig:trajectories}
we present the case of a diverging solution. Examining ODE \eqref{ode1},
we observe that for large $n$, $n$ grows quicker than $\mathcal{O}(n^2)$
and the divergence can be identified as a finite-time blow-up.
In the second case, depicted in the two plots in the middle of 
Figure~\ref{fig:trajectories},
the trajectory converges to the steady state $(0, 1)$, but does so
after entering the region $S>1$ and thereafter the
region $n<0$. Note that the steady state $(0,1)$ is still a stable node
in this case and that this overshoot is therefore not a spiralling effect.
Since these trajectories do not correspond to a non-negative solution of the
ODE system \eqref{ode1}--\eqref{ode2}, they do not represent travelling wave
solutions to the original problem. The last case, presented in the plots
on the bottom of Figure~\ref{fig:trajectories}, corresponds to an acceptable
solution and is characterised by the convergence to $(0, 1)$ without 
crossing the line $S=1$.
\begin{figure}[t]
\hskip 2cm
$\chi=1$, $c=0.5884$
\hskip 3.6cm
$\chi=0.3$, $c=0.3$ \hfill\break
\centerline{
\includegraphics[height=4.4cm]{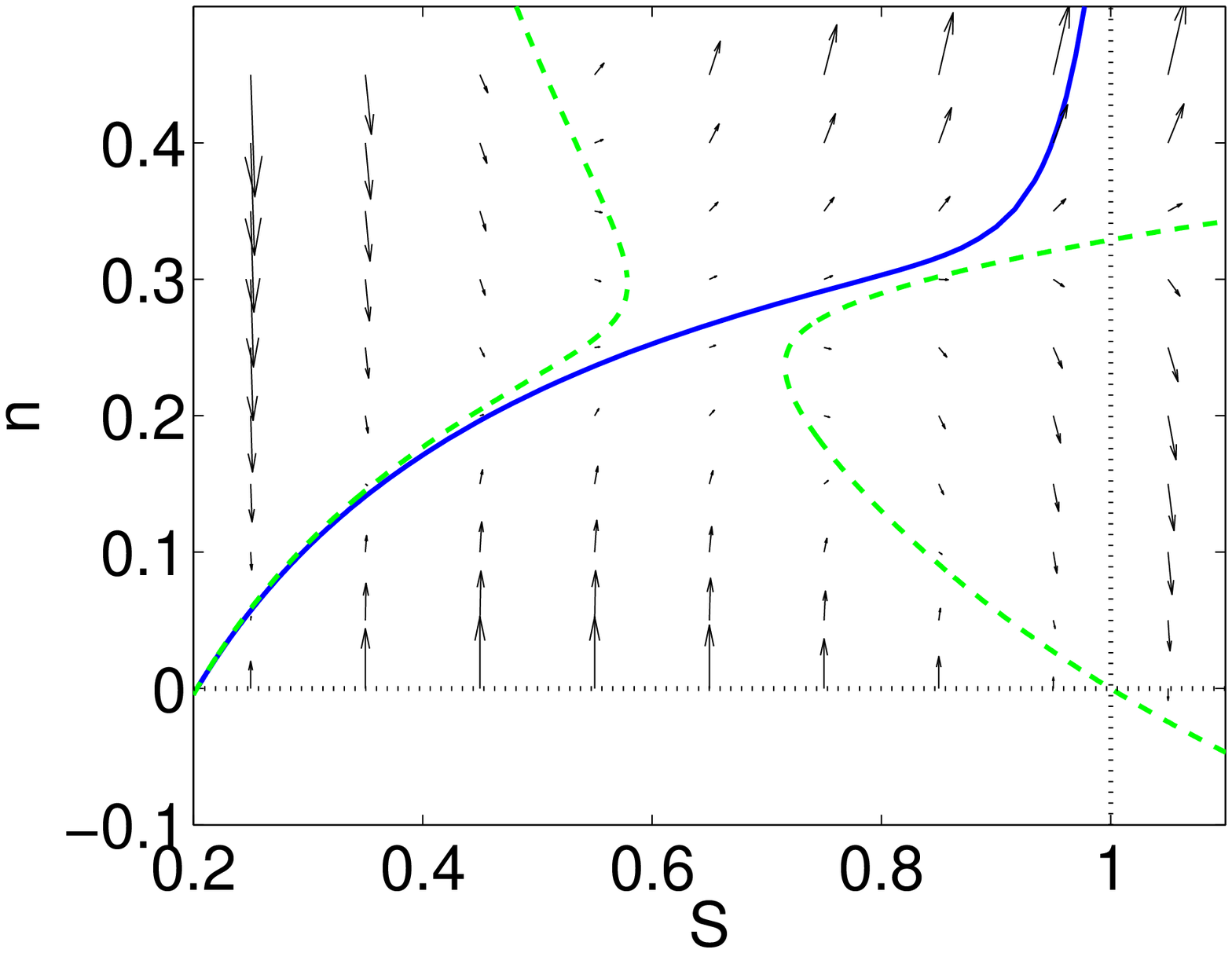}
\hskip 4mm
\includegraphics[height=4.4cm]{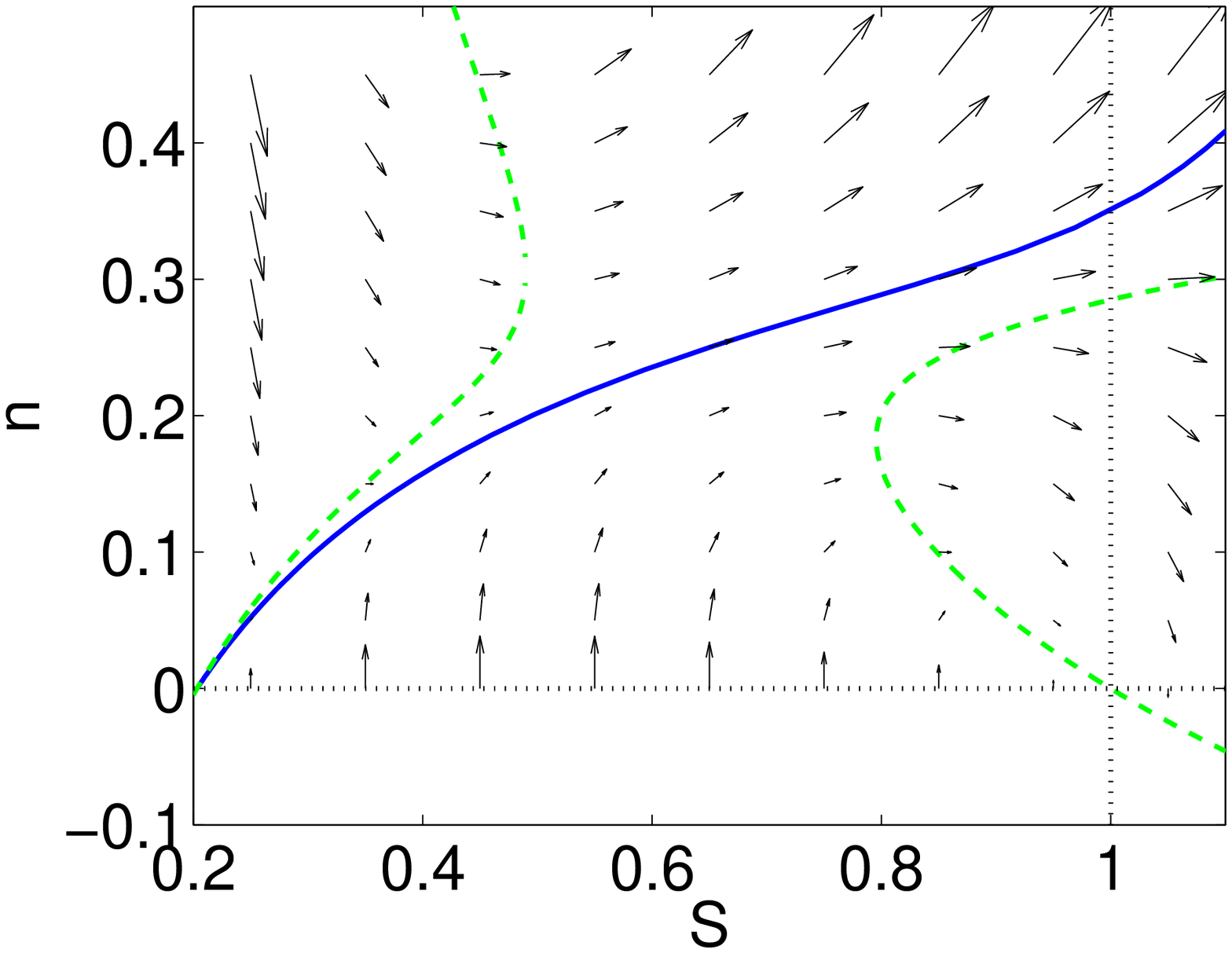}
}
\vskip 2mm
\hskip 2cm
$\chi=1$, $c=0.5885$
\hskip 3.6cm
$\chi=0.3$, $c=0.328$ \hfill\break
\centerline{
\includegraphics[height=4.4cm]{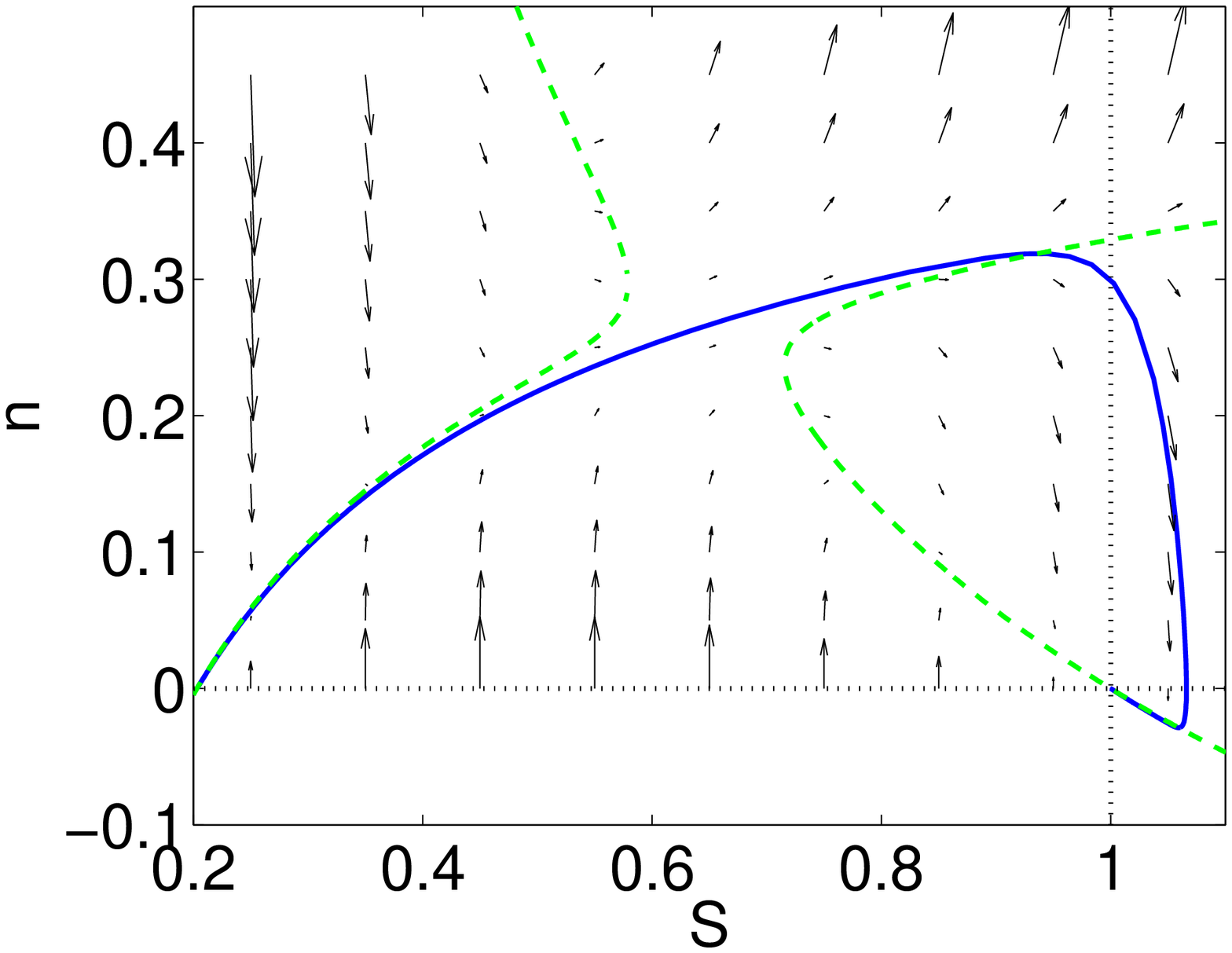}
\hskip 4mm
\includegraphics[height=4.4cm]{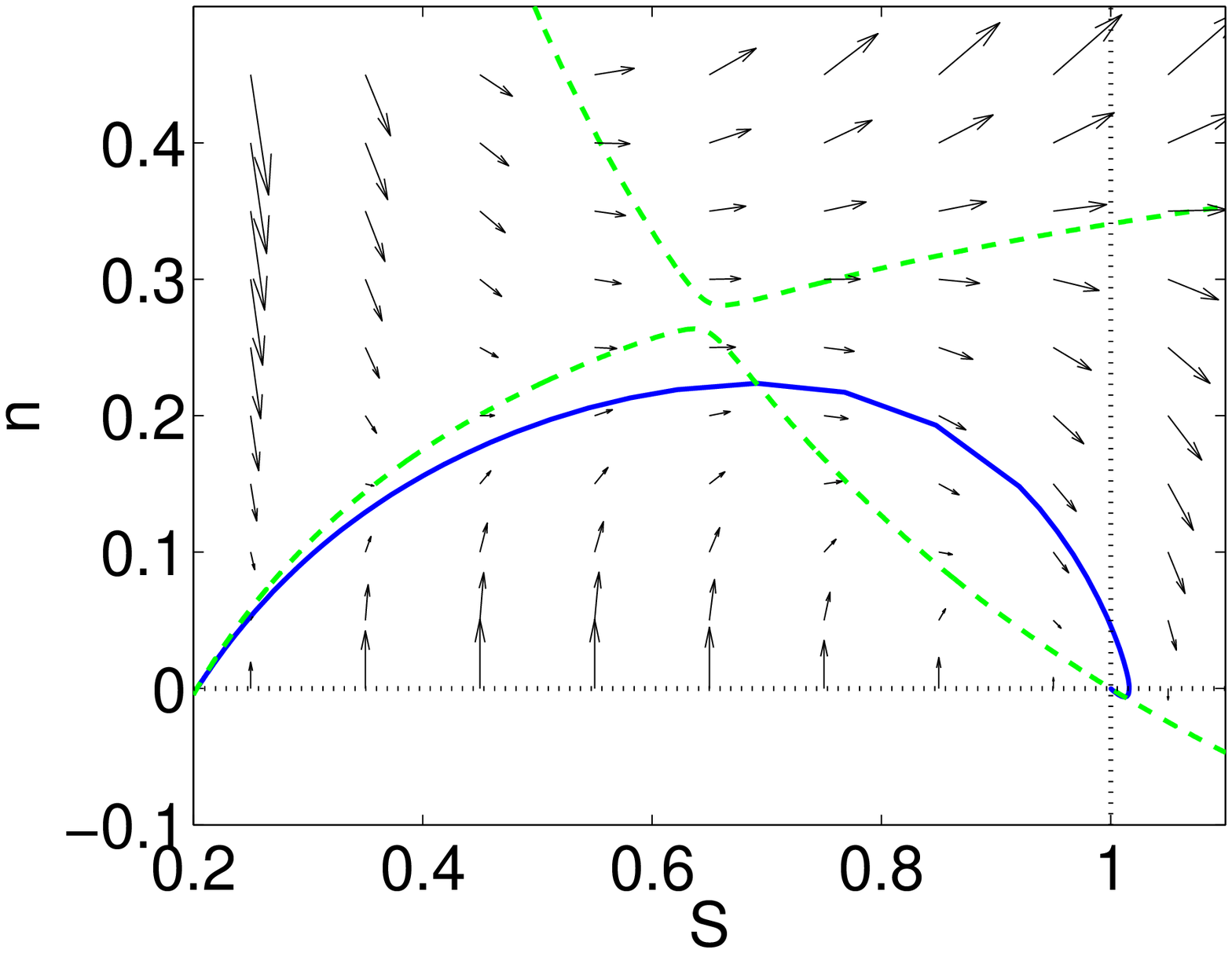}
}
\vskip 2mm
\hskip 2cm
$\chi=1$, $c=0.59$
\hskip 3.9cm
$\chi=0.3$, $c=0.35$ \hfill\break
\centerline{
\includegraphics[height=4.4cm]{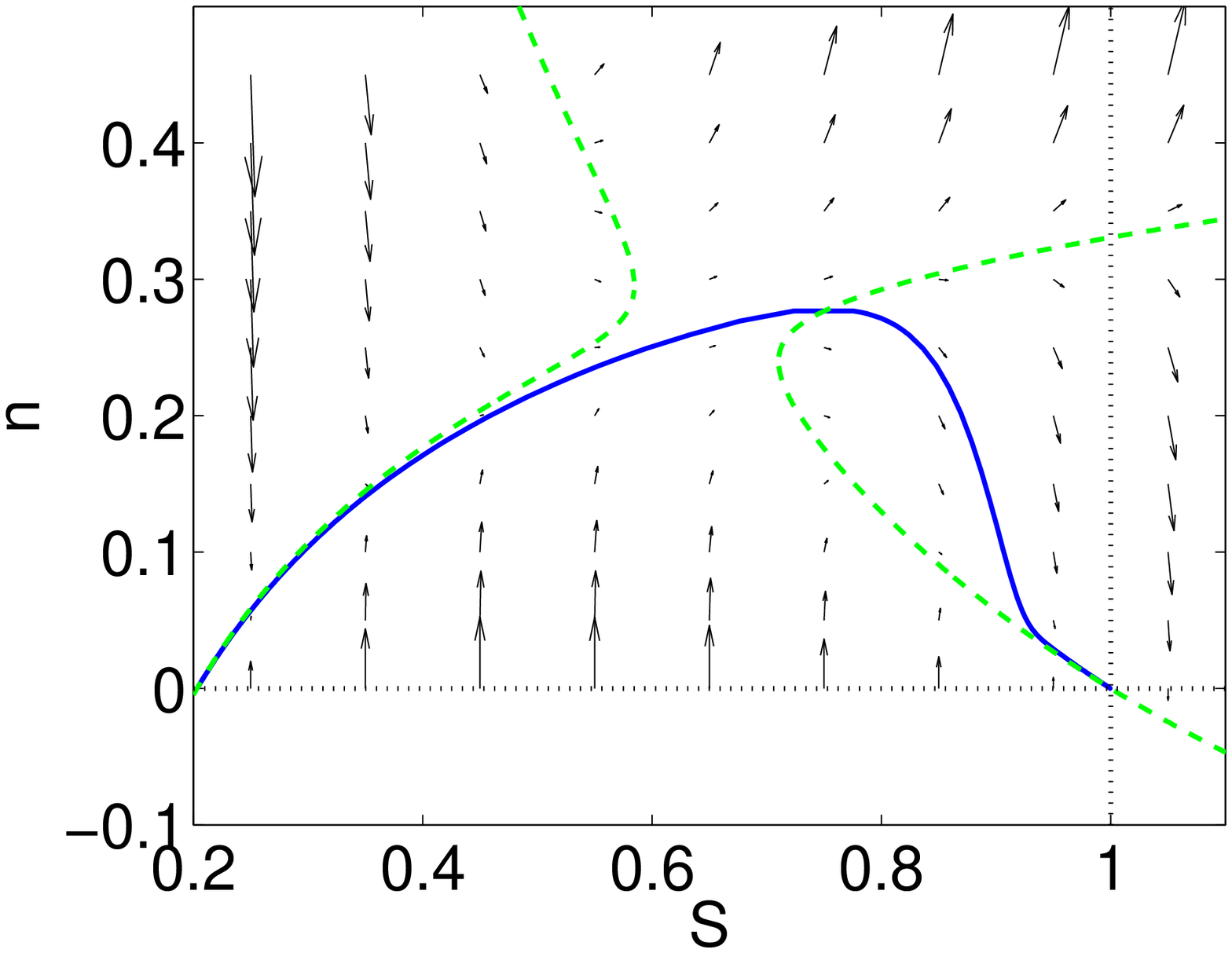}
\hskip 4mm
\includegraphics[height=4.4cm]{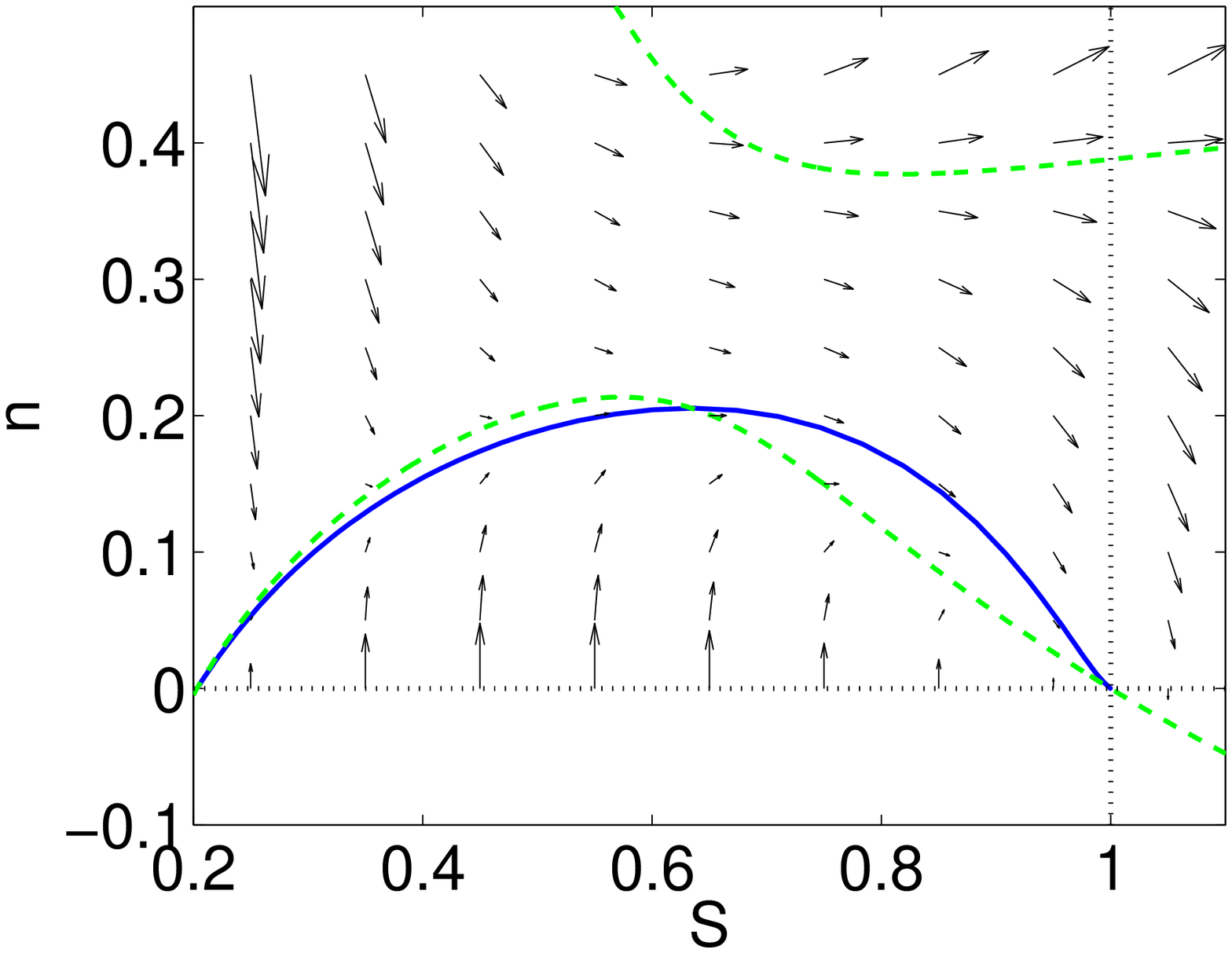}
}
\caption{{\it Trajectories of the ODE system $\eqref{ode1}$--$\eqref{ode2}$
that highlight the three different cases. Parameters in all
plots are $\lambda_0 = 10$, $S_c = 0.5$.
Solid line: trajectory, dashed line:
$n$-nullcline, dotted lines: $n=0$ and $S=1$.}}
\label{fig:trajectories}
\end{figure}

% Alternative arrangement of plots
% \begin{figure}[t]
% \hskip 1.1cm
% $\chi=1$, $c=0.5884$
% \hskip 1.7cm
% $\chi=1$, $c=0.5885$
% \hskip 1.7cm
% $\chi=1$, $c=0.59$ \hfill\break
% \centerline{
% \includegraphics[height=3.0cm]{./Figure4a.eps}
% \hskip 2mm
% \includegraphics[height=3.0cm]{./Figure4b.eps}
% \hskip 2mm
% \includegraphics[height=3.0cm]{./Figure4c.eps}
% }
% \vskip 2mm
% \hskip 1.1cm
% $\chi=0.3$, $c=0.3$
% \hskip 1.9cm
% $\chi=0.3$, $c=0.328$
% \hskip 1.7cm
% $\chi=0.3$, $c=0.35$ \hfill\break
% \centerline{
% \includegraphics[height=3.0cm]{./Figure4d.eps}
% \hskip 2mm
% \includegraphics[height=3.0cm]{./Figure4e.eps}
% \hskip 2mm
% \includegraphics[height=3.0cm]{./Figure4f.eps}
% }
% \caption{{\it Trajectories of the ODE system $\eqref{ode1}$--$\eqref{ode2}$
% that highlight the four different cases. Parameters in all
% plots are $\lambda_0 = 10$, $S_c = 0.5$.
% Solid line: trajectory, dashed line:
% $n$-nullcline, dotted lines: $n=0$ and $S=1$.}}
% \label{fig:trajectories:alt}
% \end{figure}

\subsection{Case III: Infinite chemotactic sensitivity ($\kappa = 0$)}

As $\kappa$ decreases further we observe that the minimal
wave speed necessary to allow a non-negative travelling
wave solution of \eqref{sys_pS2} increases. In the limit
$\kappa \to 0$, the ODE system \eqref{ode1}--\eqref{ode2} no longer
has convergent solutions. However, in this limit the linearisation
assumption leading to these ODEs and the system
\eqref{sys_pS2} is no longer valid and we must consider
the original turning kernel as defined in \eqref{eq:lamhybridred}.
In the limit $\kappa\to 0$ the turning rate in the hybrid model
therefore becomes
\begin{equation}
\lambda
=
\left\{ \begin{array}{rl} 0,& \qquad \mbox{for} \; y^{(1)} > 0, \\
2\lambda_0, &  \qquad \mbox{for} \; y^{(1)} < 0\,.
\end{array}\right.
\end{equation}
Hence, bacteria moving in a favourable direction do not turn, indicating 
that the
wave speed achieved in this limit should evolve to $c=s=1$.
In \cite{Xue:2011:TWH} it was shown, for a slightly different
turning kernel, that travelling waves can exist even without growth
terms and that their wave speed satisfies $c=s$.

\section{Computational analysis of the wave speed}
\label{sec:numerics}

In this section we computationally compare wave speeds from
the hybrid model with those of the fully continuous models. Specifically,
we investigate the regimes in which the latter provide an acceptable
insight into the travelling wave behaviour of the hybrid model,
and where they differ. We begin by investigating the non-chemotaxis
case, where the minimum wave speed $c^*$ for the continuum systems
was determined in (\ref{eq:cstar}). In Section~\ref{subsec:comp:chem}
we show how the wave speed depends on the value of $\kappa$, and
correspondingly the chemotactic sensitivity $\chi$ in the macroscopic model.
A comparison with hybrid models without cell proliferation is given 
in Section \ref{subsec:comp:nogrowth}. We conclude this section with 
a discussion into
the effect and origin of oscillations observed under increasing the adaptation
time $t_a$.

\subsection{Case I: No chemotaxis ($\kappa = \infty$)}
\label{subsec:comp:nochem}

In Section~\ref{subsec:analysis:nochem} we analysed
the macroscopic PDEs in the absence of chemotaxis.
Travelling wave solutions were shown to exist for all wave speeds
$c\in(c^*, 1)$, with $c^*$ determined by \eqref{eq:cstar}. In
Figure~\ref{fig:hybridSpeed:NoChem}(a), variation of
\eqref{eq:cstar} as a function of $\lambda_0$ is illustrated; we note
that wave speeds determined through simulation of the PDE systems
correspond exactly (to accuracy of the numerical approximation)
with the analytical wave speeds.
We now numerically investigate the wave speed for the case
$\chi=0$ in the hybrid model.

For our simulations we consider the same parameters and methods as
described in Section~\ref{subsec:hybrid:num}: specifically, we set the
system parameters $S_c = 0.5$, $s=1$ and $D_S = 0$. For the
computations we consider a time step $\Delta t = 10^{-3}$, a spatial
resolution of $\Delta x = 0.25$ on a domain with length $L=100$, and
simulate the system until the value of $S$ at $x=60$ falls below
$0.5$. The profiles at this time, together with the time when $S$ at 
$x=20$ falls
below $0.5$, are used to estimate the wave speed.

The measured wave speed for varying $\lambda_0$ is illustrated in
Figure~\ref{fig:hybridSpeed:NoChem}(a), along with $c^*$ as
predicted from the travelling wave analysis. While the relationship is
similar in shape, we note that at all values of $\lambda_0$ tested the
measured wave speed lies below the analytical value $c^*$. In the literature
it has been observed that inaccuracies in numerical schemes can lead 
to an increase in wave
speeds \cite{Reitz:1981:SNM}, therefore rendering
the lower wave speed seen in Figure~\ref{fig:hybridSpeed:NoChem}(a) 
as counter intuitive.

Nevertheless, we can provide the following explanation for the
distinct values in the continuum and hybrid models. For the zero-chemotaxis
case, wave generation and movement is solely determined by growth ahead
and death behind the wave. In the continuum model an outermost ``fractional 
bacteria population''
can extend significantly beyond the wave front, since some proportion of the
initial population never turns left, and hence  far into the region where 
$S$ is very close to its initial
value of $1$. Yet this fractional population still grows exponentially
($\partial p^\pm/\partial t \approx (1-S_c)p^\pm$), seeding the growth and
expansion of the population. The finite/discrete nature of the hybrid 
model precludes
any fractional bacterium: the forward ``tail'' is necessarily
finite and growth will not occur beyond the outermost individual.

For the above explanation to hold we would expect a dependence of the
measured wave speed on the initial number of bacteria $N_0$: continuous
densities provide a closer approximation under larger numbers of bacteria
and we would expect convergence in the wave speed to $c^*$. Simulations in
Figure~\ref{fig:hybridSpeed:NoChem}(b) demonstrate this property, corroborating
our interpretation.

\begin{figure}[t]
\centerline{\raise 4.3cm
\hbox{(a)}
\hskip -4mm
\includegraphics[height=4.2cm]{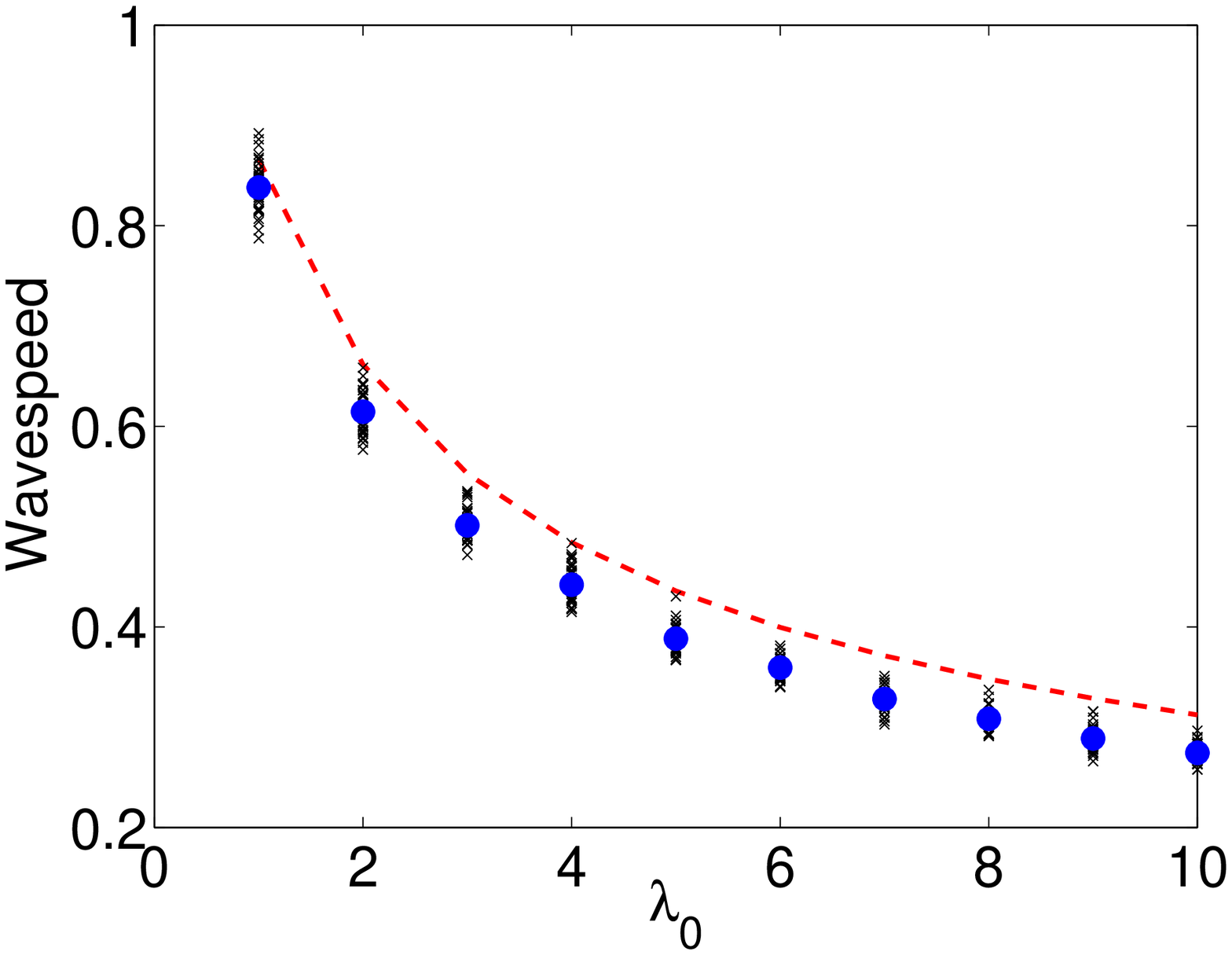}
\hskip 4mm
\raise 4.3cm \hbox{(b)}
\hskip -4mm
\includegraphics[height=4.2cm]{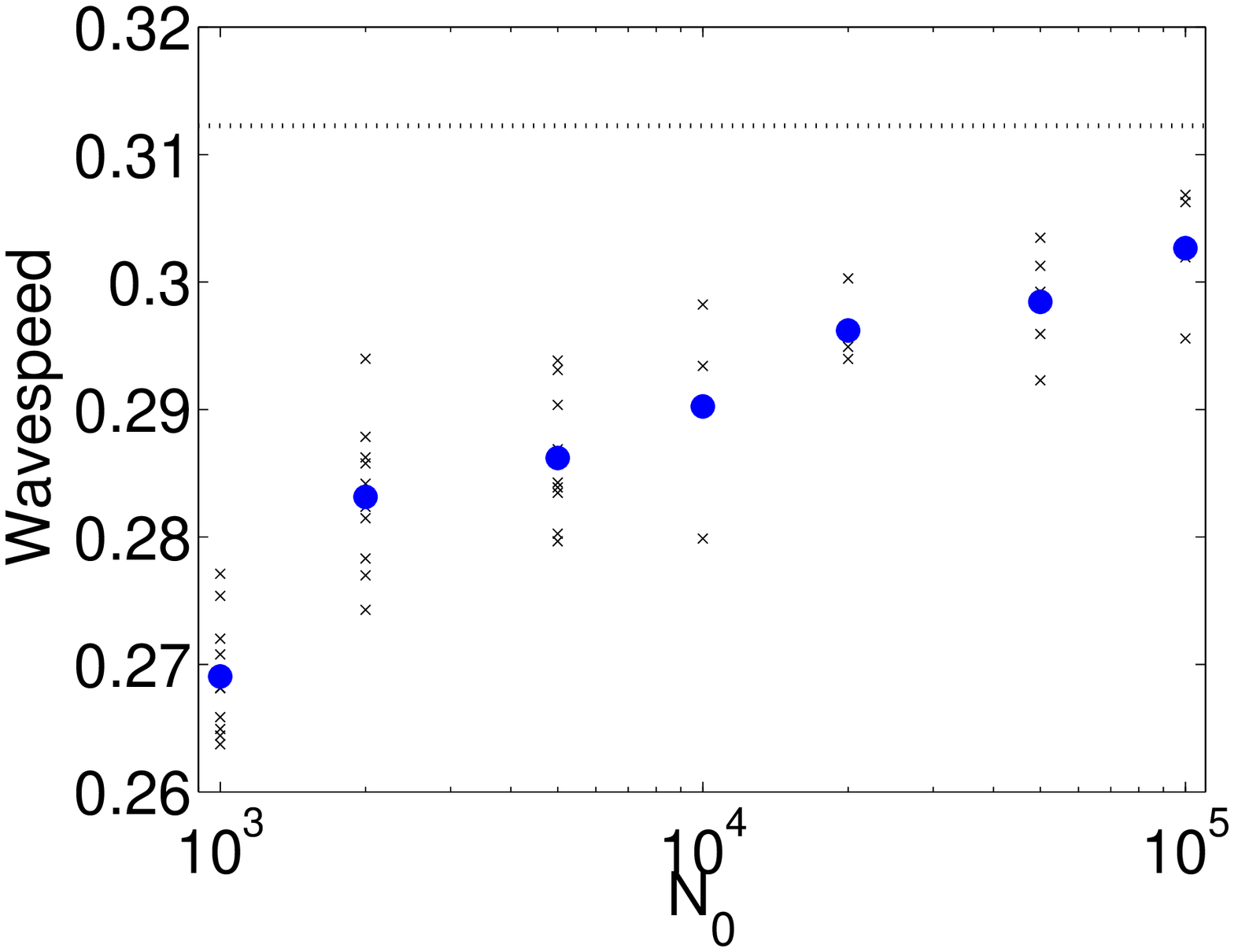}
}
\caption{{\it Measured wave speed in the hybrid model.
	Crosses: individual simulations, dots: ensemble averages.
	Parameters are as described in the text.} \hfill\break
(a) {\it Wave speed in dependence of $\lambda_0$ for $N_0 = 10,000$.
Dashed line: $c^*$ given by $(\ref{eq:cstar})$.} \hfill\break
(b) {\it Wave speed as a function of
$N_0$ with $\lambda_0 = 10$. Dashed line: $c^*$
computed by $(\ref{eq:cstar})$.}
}
\label{fig:hybridSpeed:NoChem}
\end{figure}

\subsection{  Case II: Increasing chemotaxis ($ 0 < \kappa < \infty$) }
\label{subsec:comp:chem}

In the second set of numerical experiments we measure the dependency of the
wave speed on the critical parameter $\kappa$, i.e. we determine the effect
of increasing chemotaxis as $\kappa$ decreases. We compare the results
measured for the hybrid system with the continuous Systems \textbf{(A)} 
and \textbf{(B)}.

We use the same parameters as in Section~\ref{subsec:comp:nochem} and
results are shown in Figure~\ref{fig:hybridSpeed:Chem}. The results
demonstrate the regimes where correspondence across the varying modelling
levels occurs: while the hybrid model (dotted line) corresponds well with its
closest continuous version (mesoscopic System \textbf{(A)}, red solid line) 
over
a wide range of $\kappa$, it only corresponds with System \textbf{(B)} 
(black dashed line)
for larger $\kappa$, diverging as $\kappa$ decreases. Note that the
turning rate \eqref{eq:lamcont} used for System \textbf{(B)} becomes
negative at small values of $\kappa$ and we limit the range of $\kappa$ 
studied accordingly.

At larger $\kappa$ all three models converge to a value close to $c^*$ 
as $\kappa$ grows:
in this regime the main assumption proposed for the linearisation
($|S(x) - y^{(2)}| \ll \kappa$) holds and we obtain good quantitative 
agreement. While this
assumption becomes less acceptable as we decrease $\kappa$, leading to 
the divergent
behaviour described above, we note that all models show the same qualitative 
agreement:
increasing chemotactic responses leads to an increase in the wave speed.
Note that the results for System \textbf{(B)} can be identically
replicated using the ODE system \eqref{ode1}--\eqref{ode2} and
a search algorithm for the smallest value of $c$ that admits a nonnegative 
solution to the system.

These numerical experiments demonstrate that chemotaxis has a
significant effect on the speed of movement and that the waves
cannot solely be explained by growth and death terms. Rather, we
interpret birth and death processes as stabilisers to what would 
otherwise be
transient waves \cite{Franz:2012:HMI,Xue:2011:TWH}. This interpretation 
is in agreement with the
results presented in Figure~\ref{fig:nogrowth}, as the initial wave speed 
for the system without growth
seems to be similar to the wave speed of the system including growth 
and death terms.

\begin{figure}[t]
\centerline{
\includegraphics[height=6cm]{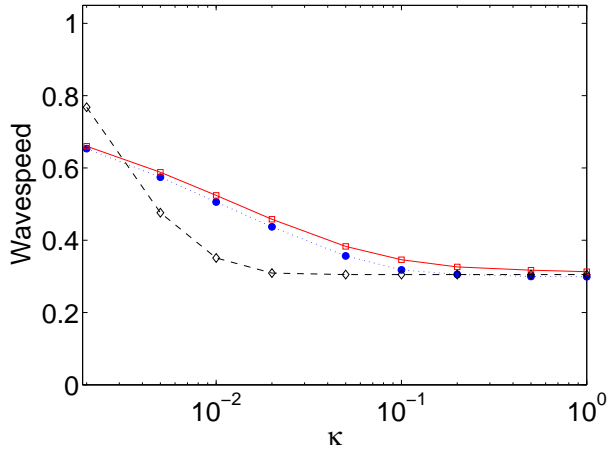}
}
\caption{{\it Comparison between wave speeds of the
	various models in dependence of $\kappa$.
	Dotted line: hybrid model, red solid line: mesoscopic
	System~\textbf{(A)}, dashed line: linearised System~\textbf{(B)}.
	Parameters are as described in the text.}
}
\label{fig:hybridSpeed:Chem}
\end{figure}

\subsection{Oscillations in the wave speed}
\label{subsec:comp:oscillations}
An additional observation we made during the numerical
experiments of the hybrid model is that for increasing values of the
adaptation time $t_a$, the wave speed starts to differ strongly
from the mesoscopic System \textbf{(A)}, an effect that we identified to be due
to oscillations in the behaviour of the wave. In
Figure~\ref{fig:oscillations}(a) we present an example of strongly
oscillating wave speeds (where the wave speed is
measured as rate of change of the average position of bacteria).
This example occurred for the parameters $S_c = 0.5$, $\lambda_0 = 10$,
$\kappa = 0.001$ and $t_a = 4$. We can also clearly see that the
wave speed is correlated to the current number of agents in the system.
In the literature similar effects of oscillating waves in stochastic models
have been observed \cite{Metcalf:1994:OWF,Othmer:1998:OCS}.

In Figure~\ref{fig:oscillations}(b), we present the form of the wave
at different times throughout the simulation. It is clearly visible that
the shape differs significantly at different times. One reason these 
oscillations
occur when $t_a$ is very high is that
a bacterium that happens to be in front of the wave experiences
a very high value of $S$, whilst its internal dynamics only adapt
very slowly. This, in combination with the low value of $\kappa$,
leads to a bacterium that does not switch direction for a long time
and will proliferate at a high rate. This implies that a spike of
bacteria forms in front of the wave that moves faster than the rest of
the wave. We can clearly see such a spike in the left-most waveform in
Figure~\ref{fig:oscillations}(b). Once the frontrunning bacterium and 
its copies
have turned, the wave goes into a reordering phase (second and third 
waveform),
until, eventually, a new spike emerges (4th waveform).

In Figure~\ref{fig:oscillations}(c) we plot the wave speed over time for 
a smaller
value of $t_a$. We can see that the oscillations are less severe and more 
frequent
than in Figure~\ref{fig:oscillations}(a), which is in agreement with the 
explanation
above. As we decrease $t_a$ the frontrunning bacteria will adapt quicker to
their surroundings and are thereby more likely to turn. We show the influence
of changing $N_0$ on the oscillating behaviour in 
Figure~\ref{fig:oscillations}(d).
The oscillations seem to occur with a similar frequency but more regular 
to those before,
which can be explained by the increased likelihood of frontrunning bacteria 
with a higher number of agents and reduced noise in the system.

\begin{figure}
\centerline{\raise 4.3cm
\hbox{(a)}
\hskip -4mm
\includegraphics[height=4.2cm]{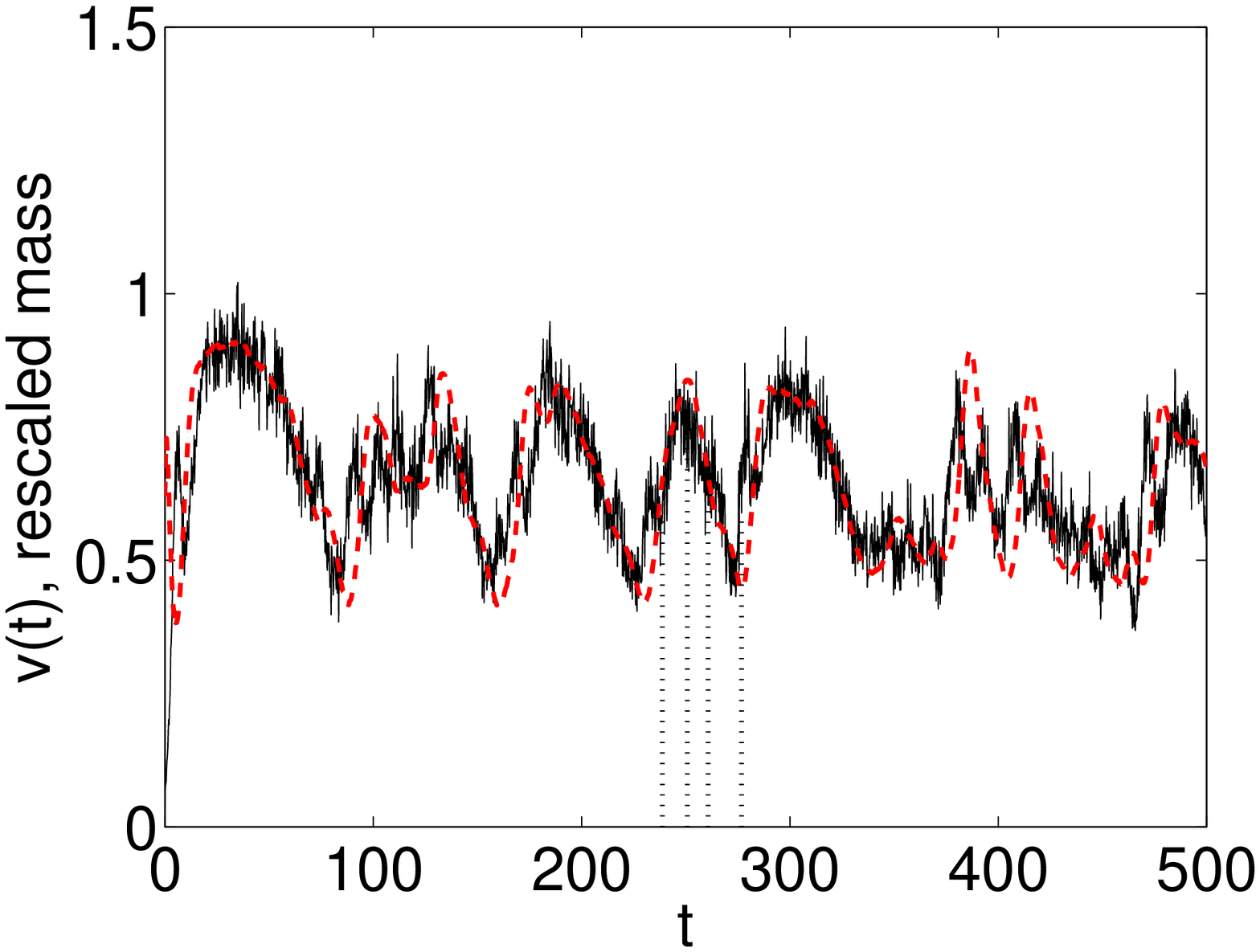}
\hskip 4mm
\raise 4.3cm \hbox{(b)}
\hskip -4mm
\includegraphics[height=4.2cm]{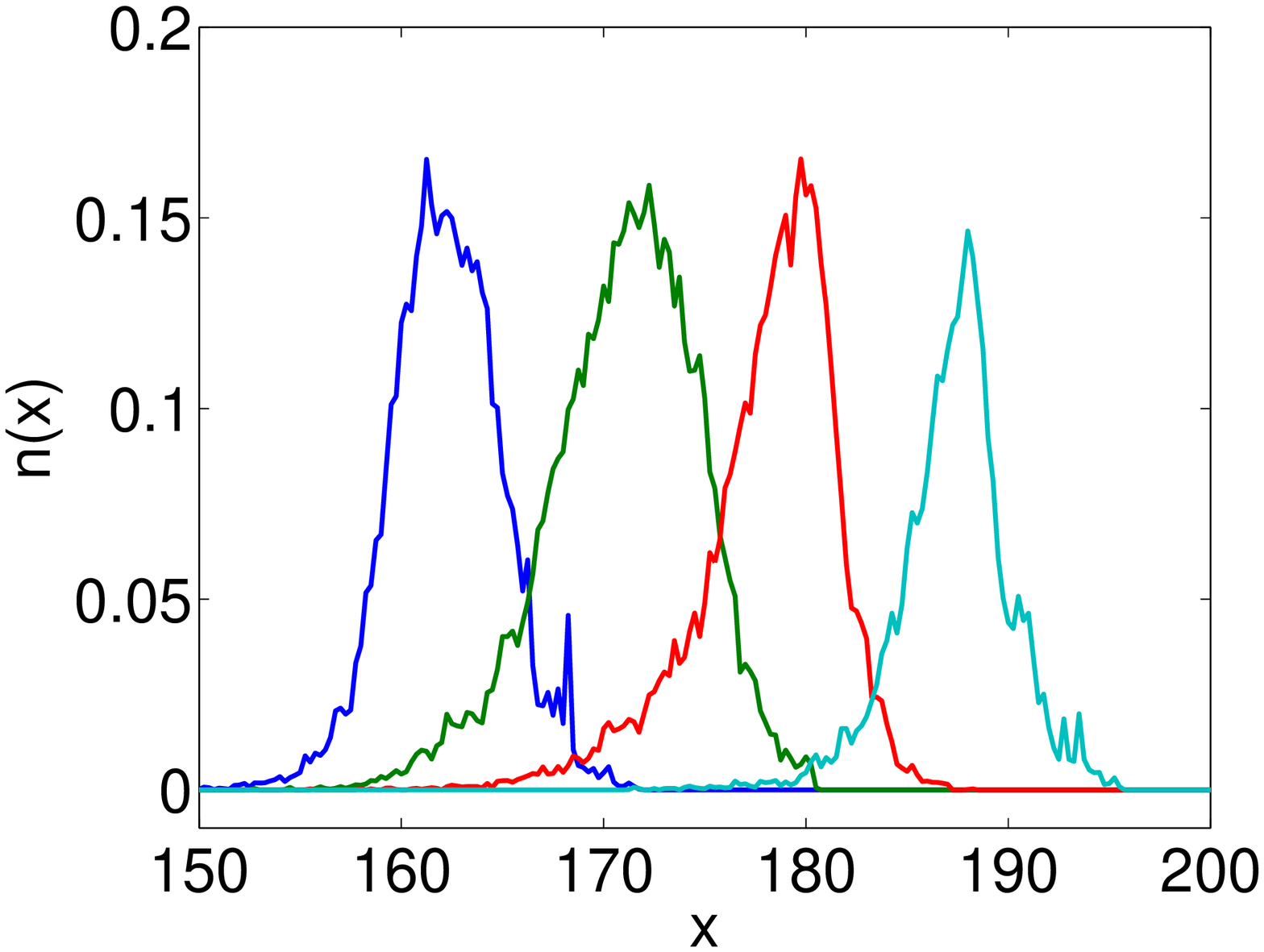}
}
\centerline{\raise 4.3cm
\hbox{(c)}
\hskip -4mm
\includegraphics[height=4.2cm]{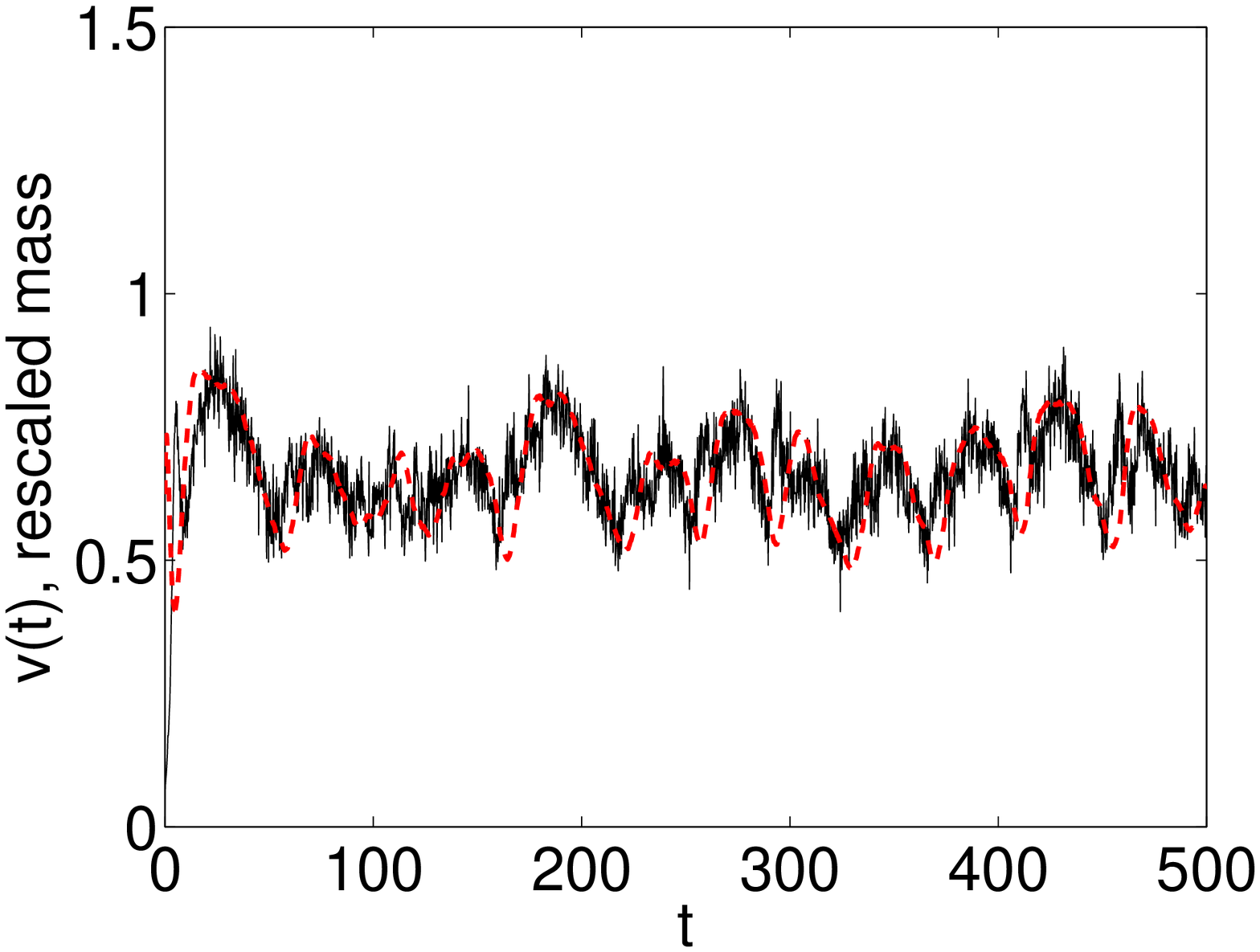}
\hskip 4mm
\raise 4.3cm \hbox{(d)}
\hskip -4mm
\includegraphics[height=4.2cm]{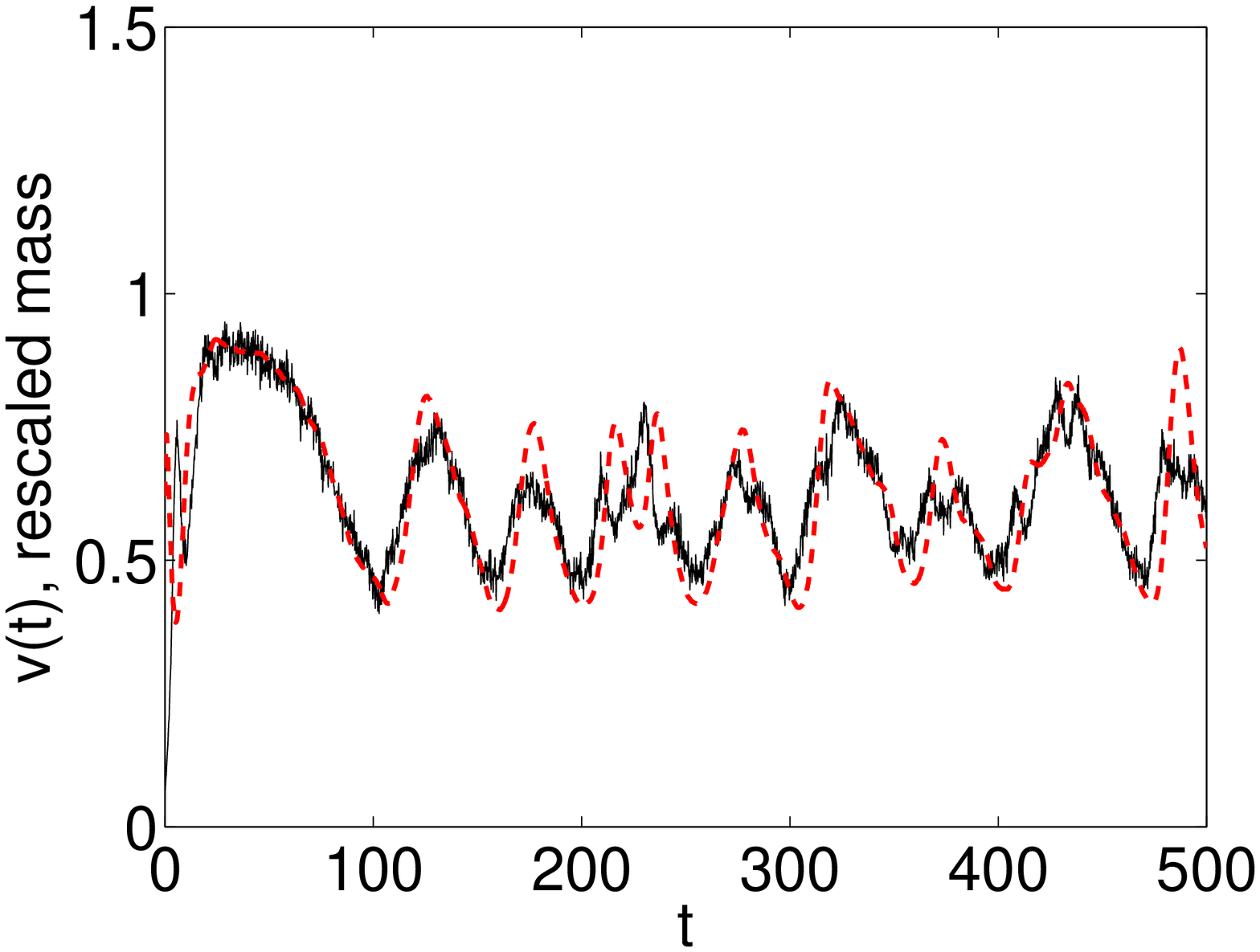}
}

\caption{{\it Oscillations in the wave speed
of the hybrid model $(\ref{eq:hybridmodel:x})$--$(\ref{eq:hybridmodel:y2})$
and $(\ref{eq:SforwardEuler})$.} \hfill\break
(a) {\it Wave speed in comparison to current number of particles 
for $N_0 = 10,000$, $t_a = 4$.
Solid line: wave speed, dashed line: number of particles, dotted 
lines: times of wave forms shown in panel (b).}\hfill\break
(b) {\it Waveform at 4 distinct times marked in panel (a) 
from left to right.} \hfill\break
(c) {\it As in (a) with $N_0 = 10,000$, $t_a = 2$.}\hfill\break
(d) {\it As in (a) with $N_0 = 50,000$, $t_a = 4$.}\hfill\break
{\it Other parameters are given in Section \ref{subsec:comp:oscillations}.}
}
\label{fig:oscillations}
\end{figure}

\section{Discussion}
\label{sec:discussion}

In this paper we presented a hybrid model of chemotaxis, incorporating
a biologically realistic turning kernel introduced in
\cite{Xue:2011:TWH}. We analysed the travelling wave behaviour of this
hybrid system using mesoscopic and macroscopic equations, deriving an analytical
value for the expected wave speed in the case of no chemotaxis. As chemotaxis
increases we demonstrated (analytically and numerically) that the expected
wave speed increases, indicating that the wave that forms is not solely driven by
growth and death processes. In contrast to the transient waves observed for the
hybrid model in the absence of growth and death terms \cite{Franz:2012:HMI},
the (numerical) waves observed here in their presence are stable, indicating
the stabilising effect of birth and death. The numerical analysis
reveals that the macroscopic equations derived through linearisation of the turning
kernel can qualitatively describe the change in wave speed as chemotaxis
increases, but that there are significant quantitative differences between the two
systems. Additionally, we observed oscillations in the wave movement, an effect that
had been seen in similar systems in the literature \cite{Othmer:1998:OCS} and that
cannot be explained using mean-field approximations. 

To date, travelling waves in chemotaxis models have mainly been
analysed from the perspective of macroscopic PDE models of chemotaxis
\cite{Horstmann:2004:UPK,Hillen:2009:UGP}. The existence of
travelling waves for continuum models with growth terms is well established
\cite{Satnoianu:2001:TWN,Nadin:2008:TWK,Landman:2003:CCM}. While
hybrid models have been used to study pattern formation in bacterial chemotaxis
\cite{Guo:2008:HAB,Xue:2011:RSS}, these studies have not analysed the 
travelling wave
patterns observed in bacterial cell populations.

Recently, experimental studies using microfluidic techniques tracked cell trajectories within 
a traveling pulse, and revealed that persistence of direction in cell movement accounts for 
30\% of the macroscopic speed of the traveling pulse \cite{Saragosti:2011:DPC}. The hybrid model framework 
studied here provides a natural method for direct comparison of model predictions with experimental measurements of cell 
trajectory, and 
this is left as future work.

\section*{Acknowledgements}

The research leading to these results has received funding from
the European Research Council under the {\it European Community's}
Seventh Framework Programme {\it (FP7/2007-2013)} /
ERC {\it grant agreement} No. 239870. This publication
was based on work supported in part by Award No KUK-C1-013-04, made by King
Abdullah University of Science and Technology (KAUST).
Radek Erban would also like to thank the Royal Society for
a University Research Fellowship; Brasenose College, University of Oxford,
for a Nicholas Kurti Junior Fellowship; and the Leverhulme Trust for
a Philip Leverhulme Prize. This prize money
was used to support research visits of Chuan Xue and Kevin Painter
in Oxford. Kevin Painter acknowledges a Leverhulme Trust Research
Fellowship award (RF-2011-045).


\begin{thebibliography}{10}

\bibitem{Adler:1966:CB}
J.~Adler.
\newblock Chemotaxis in bacteria.
\newblock {\em Science}, 153:708--716, 1966.

\bibitem{Berg:1975:HBS}
H.~Berg.
\newblock How bacteria swim.
\newblock {\em Scientific American}, 233:36--44, 1975.

\bibitem{Berg:1972:CEC}
H.~Berg and D.~Brown.
\newblock Chemotaxis in {E}sterichia coli analysed by three-dimensional
  tracking.
\newblock {\em Nature}, 239:500--504, 1972.

\bibitem{Bourret:1991:STP}
R.~Bourret, K.~Borkovich, and M.~Simon.
\newblock Signal transduction pathways involving protein phosphorylation in
  prokaryotes.
\newblock {\em Annual Review of Biochemistry}, 60:401--441, 1991.

\bibitem{Brenner:1998:PMC}
M.~Brenner, L.~Levitov, and E.~Budrene.
\newblock Physical mechanisms for chemotactic pattern formation by bacteria.
\newblock {\em Biophysical Journal}, 74(4):1677--1693, 1998.

\bibitem{Budrene:1991:CPF}
E.~Budrene and H.~Berg.
\newblock Complex patterns formed by motile cells of {E}sterichia coli.
\newblock {\em Nature}, 349:630--633, February 1991.

\bibitem{Budrene:1995:DFS}
E.~Budrene and H.~Berg.
\newblock Dynamics of formation of symmetrical patterns by chemotactic
  bacteria.
\newblock {\em Nature}, 376:49--53, July 1995.

\bibitem{Cercignani:1994:MTD}
C.~Cercignani, R.~Illner, and M.~Pulvirenti.
\newblock {\em {T}he {M}athematical {T}heory of {D}ilute {G}ases}.
\newblock Applied Mathematical Sciences, 106, Springer-Verlag, 1994.

\bibitem{Chavanis:2010:SKS}
P.~Chavanis.
\newblock A stochastic {K}eller-{S}egel model of chemotaxis.
\newblock {\em Communications in nonlinear science and numerical simulations},
  15:60--70, 2010.

\bibitem{Erban:2005:ICB}
R.~Erban.
\newblock {\em From individual to collective behaviour in biological systems}.
\newblock PhD thesis, University of Minnesota, 2005.

\bibitem{Erban:2012:ICB}
R.~Erban and J.~Haskovec.
\newblock From individual to collective behaviour of coupled velocity jump
  processes: A locust example.
\newblock {\em Kinetic and Related Models}, 5(4):817--842, 2012.

\bibitem{Erban:2004:ICB}
R.~Erban and H.~Othmer.
\newblock From individual to collective behaviour in bacterial chemotaxis.
\newblock {\em SIAM Journal on Applied Mathematics}, 65(2):361--391, 2004.

\bibitem{Erban:2005:STS}
R.~Erban and H.~Othmer.
\newblock From signal transduction to spatial pattern formation in {{\em E.
  coli}}: A paradigm for multi-scale modeling in biology.
\newblock {\em Multiscale Modeling and Simulation}, 3(2):362--394, 2005.

\bibitem{Fisher:1937:WAA}
R.~Fisher.
\newblock The wave of advance of advantageous genes.
\newblock {\em Annals of Eugenics}, 7:355--369, 1937.

\bibitem{Franz:2012:HMI}
B.~Franz and R.~Erban.
\newblock Hybrid modelling of individual movement and collective behaviour.
\newblock In M.~Lewis, P.~Maini, and S.~Petrovskii, editors, {\em Dispersal,
  individual movement and spatial ecology: A mathematical perspective}.
  Springer, 2013.

\bibitem{Gerisch:MMC:2010}
A.~Gerisch and K.~Painter.
\newblock Mathematical modelling of cell adhesion and its applications to
  developmental biology and cancer invasion.
\newblock In A.~Chauviere and L.~Preziosi, editors, {\em Cell Mechanics: From
  Single Scale-Based Models to Multiscale Modeling}, Chapter~12, pages
  319--350. CRC Press, 2010.

\bibitem{Guo:2008:HAB}
Z.~Guo, P.~Sloot, and J.~Tay.
\newblock A hybrid agent-based approach for modeling microbiological systems.
\newblock {\em Journal of Theoretical Biology}, 255:163--175, 2008.

\bibitem{Hillen:2009:UGP}
T.~Hillen and K.~Painter.
\newblock A user's guide to pde models for chemotaxis.
\newblock {\em Journal of Mathematical Biology}, 58:183--217, 2009.

\bibitem{Horstmann:2004:UPK}
D.~Horstmann.
\newblock From 1970 until present: the {K}eller-{S}egel model in chemotaxis and
  its consequences II.
\newblock {\em Jahresbericht des Deutschen Mathematiker Vereins}, 106:51--69,
  2004.

\bibitem{Keller:1971:TBC}
E.~Keller and L.~Segel.
\newblock Traveling bands of chemotactic bacteria: A theoretical analysis.
\newblock {\em Journal of Theoretical Biology}, 30:235--248, 1971.

\bibitem{Kennedy:1980:TWS}
C.~Kennedy and R.~Aris.
\newblock Traveling waves in a simple population model involving growth and
  death.
\newblock {\em Bulletin of Mathematical Biology}, 42:397--429, 1980.

\bibitem{Landman:2003:CCM}
K.~Landman, G.~Petter, and D.~Newgreen.
\newblock Chemotactic cellular migration: smooth and discontinuous travelling
  wave solutions.
\newblock {\em SIAM Journal on Applied Mathematics}, 63(5):1666--1681, 2003.

\bibitem{Landman:2005:DCD}
K.~Landman, M.~Simpson, J.~Slater, and D.~Newgreen.
\newblock Diffusive and chemotactic cellular migration: Smooth and
  discontinuous travelling wave solutions.
\newblock {\em SIAM Journal of Applied Mathematics}, 65:1420--1442, 2005.

\bibitem{Landman:2007:MEI}
K.A. Landman, M.J. Simpson, and D.F. Newgreen.
\newblock Mathematical and experimental insights into the development of the
  enteric nervous system and hirschsprung's disease.
\newblock {\em Development, growth and differentiation}, 49:277--286, 2007.

\bibitem{Li:2011:ANS}
T.~Li and Z.~Wang.
\newblock Asymptotic nonlinear stability of traveling waves to conservation
  laws arising from chemotaxis.
\newblock {\em Journal of Differential Equations}, 250:1310--1333, 2011.

\bibitem{Li:2012:SPW}
T.~Li and Z.~Wang.
\newblock Steadily propagating waves of a chemotaxis model.
\newblock {\em Mathematical Biosciences}, 240:161--168, 2012.

\bibitem{Lui:2010:TWS}
R.~Lui and Z.~Wang.
\newblock Traveling wave solutions from microscopic to macroscopic chemotaxis
  models.
\newblock {\em Journal of Mathematical Biology}, 61:739--761, 2010.

\bibitem{Metcalf:1994:OWF}
M.~Metcalf, J.~Merkin, and S.~Scott.
\newblock Oscillating wave fronts in isothermal chemical systems with arbitrary
  powers of autocatalysis.
\newblock {\em Proceedings of the Royal Society London A}, 447:155--174, 1994.

\bibitem{Murray:2002:MB}
J.~Murray.
\newblock {\em {M}athematical {B}iology}.
\newblock Springer Verlag, 2002.

\bibitem{Nadin:2008:TWK}
G.~Nadin, B.~Perthame, and L.~Ryzhik.
\newblock Traveling waves for the {K}eller-{S}egel system with {F}isher birth
  term.
\newblock {\em Interfaces and Free Boundaries}, 10:517--538, 2008.

\bibitem{Othmer:1988:MDB}
H.~Othmer, S.~Dunbar, and W~Alt.
\newblock Models of dispersal in biological systems.
\newblock {\em Journal of Mathematical Biology}, 26:263--298, 1988.

\bibitem{Othmer:1998:OCS}
H.~Othmer and P.~Schaap.
\newblock Oscillatory c{AMP} signaling in the development of {D}ictyostelium
  discoideum.
\newblock {\em Comments on Theoretical Biology}, 5:175--282, 1998.

\bibitem{Reitz:1981:SNM}
R.D. Reitz.
\newblock A study of numerical methods for reaction-diffusion equations.
\newblock {\em SIAM Journal on Scientific and Statistical Computing},
  2:95--106, 1981.

\bibitem{Saragosti:2010:MDB}
J.~Saragosti, V.~Calvez, N.~Bournaveas, A.~Buguin, P.~Silberzan, and
  B.~Perthame.
\newblock Mathematical description of bacterial traveling pulses.
\newblock {\em PLoS Computational Biology}, 6:e1000890, 2010.

\bibitem{Saragosti:2011:DPC}
J.~Saragosti, V.~Calvez, N.~Bournaveas, B.~Perthame, A.~Buguin, and
  P.~Silberzan.
\newblock Directional persistence of chemotactic bacteria in a traveling
  concentration wave.
\newblock {\em Proceedings of the National Academy of Sciences},
  108(39):16235--16240, 2011.

\bibitem{Satnoianu:2001:TWN}
R.~Satnoianu, P.~Maini, F.~Garduno, and J.~Armitage.
\newblock Travelling waves in a nonlinear degenerate diffusion model for
  bacterial pattern formation.
\newblock {\em Discrete and Continuous Dynamical Systems B}, 1:339--362, 2001.

\bibitem{Wang:2013:MTW}
Z.A. Wang.
\newblock Mathematics of traveling waves in chemotaxis -- review paper.
\newblock {\em Discrete and Continuous Dynamical Systems Series B},
  13:601--641, 2013.

\bibitem{Witte:1997:GPW}
M.~B. Witte and A.~Barbul.
\newblock General principles of wound healing.
\newblock {\em Surgical Clinics of North America}, 77:509--528, 1997.

\bibitem{Xue:2011:RSS}
C.~Xue, E.~Budrene, and H.~Othmer.
\newblock Radial and spiral streams in {\it {p}roteus mirabilis} colonies.
\newblock {\em PLoS Computational Biology}, 7(12):e1002332, 2011.

\bibitem{Xue:2011:TWH}
C.~Xue, H.~Hwang, K.~Painter, and R.~Erban.
\newblock Travelling waves in hyperbolic chemotaxis equations.
\newblock {\em Bulletin of Mathematical Biology}, 73(8):1695--1733, 2011.

\bibitem{Xue:2009:MMT}
C.~Xue and H.~Othmer.
\newblock Multiscale models of taxis-driven patterning in bacterial
  populations.
\newblock {\em SIAM Journal on Applied Mathematics}, 70(1):133--167, 2009.

\end{thebibliography}
\end{document}